%% file: main.tex
\documentclass[a4paper]{article}
\usepackage[utf8]{inputenc}
\usepackage{amssymb, amsmath, amsthm, amsfonts}
\usepackage{xcolor,pgf,tikz,pgflibraryarrows,pgffor,pgflibrarysnakes,pgflibraryshapes}
\usepackage{latexsym}
\usepackage{url,xspace}
\usepackage{stmaryrd}
\usepackage{enumerate}
\usepackage{yhmath,paralist}
\usepackage{hyperref}
\usepackage[vlined,linesnumbered,ruled]{algorithm2e}
\usepackage{graphicx}
\usepackage{caption}
\usepackage{subcaption}
\usepackage{balance}
\usepackage[disable]{todonotes}
\usepackage{booktabs}
\usepackage{authblk}
\usepackage{microtype}

\newcommand{\game}{G}
\newcommand{\tuple}[1]{\ensuremath{\left\langle#1\right\rangle}}
\newcommand{\outcome}[1]{\ensuremath{\mathsf{Outcome}\left(#1\right)}}

\def\Succname{\ensuremath{{\sf Succ}}}
\def\Succ#1{\Succname\left(#1\right)}

\def\bad{\ensuremath{{\sf Bad}}}

\def\win{\ensuremath{{\sf Win}}}
\def\losing{\ensuremath{{\sf Lose}}}
\def\Attr{\ensuremath{{\sf Attr}}}

\def\uc#1{\ensuremath{\uparrow\!\left(#1\right)}}
\def\uco#1#2{\ensuremath{\uparrow^{#2}\!\left(#1\right)}}

\def\minac#1{\ensuremath{\left\lfloor #1\right\rfloor}}

\def\PS{\ensuremath{\mathcal{S}}\xspace}
\def\PR{\ensuremath{\mathcal{T}}\xspace} 

\def\ALS{\ensuremath{{\sf AntiLosing}}}

\def\Frontier{\ensuremath{{\sf Frontier}}}

\def\nat{\textsc{nat}\xspace}
\def\rct{\textsc{rct}\xspace}
\def\states{\ensuremath{\textsc{States}(\tau)}\xspace}
\def\act{\textsc{Active}}
\def\eligible{\textsc{Eligible}}
\def\laxity{\textsc{Laxity}}
\def\succS{Succ_\PS}
\def\succR{Succ_\PR}

\def\wbs{\ensuremath{\operatorname{\unrhd}}}
\def\nwbs{\ensuremath{\operatorname{\not\!\!\unrhd}}}

\def\iesim{\sqsupseteq}

\def\ES{\ensuremath{{\textsf{ES}}}\xspace}

\def\Fw{\ensuremath{{\textsf{OTFUR-TBA}}}\xspace}
\def\Bw{\ensuremath{{\textsf{BW-TBA}}}\xspace}
\def\EDF{\ensuremath{{\textsf{EDF}}}\xspace}

\def\Prename{\ensuremath{{\sf Pre}}}
\def\PreEname{\exists\Prename}
\def\PreAname{\forall\Prename}
\def\PreA#1{\PreAname\left(#1\right)}
\def\PreE#1{\PreEname\left(#1\right)}
\def\UPrename{\ensuremath{{\sf UPre}}}
\def\UPre#1{\UPrename\left(#1\right)}

\def\UPreATCname{\UPrename^\#}
\def\UPreATC#1{\UPreATCname\left(#1\right)}
\def\ATC{\mathcal{A}}

\def\PreAATCname{\PreAname^\#}
\def\PreEATCname{\PreEname^\#}
\def\PreAATC#1{\PreAATCname \left(#1\right)}
\def\PreEATC#1{\PreEATCname \left(#1\right)}

\def\PreAS#1#2{\forall \Prename^{*} \left(#1,#2\right)}
\def\UPreS#1#2{\UPrename^{*} \left(#1,#2\right)}

\def\Prename{\ensuremath{{\sf Pre}}}

\newtheorem{definition}{Definition}
\newtheorem{example}{Example}
\newtheorem{lemma}{Lemma}
\newtheorem{theorem}{Theorem}
\newtheorem{proposition}{Proposition}

\begin{document}

\title{A Backward Algorithm for the Multiprocessor Online
  Feasibility of Sporadic Tasks\footnote{This work has been supported by the
F.R.S./FNRS PDR grant FORESt: Formal verificatiOn techniques for
REal-time Scheduling problems.}}

\author{Gilles Geeraerts and Jo\"el Goossens and Thi-Van-Anh Nguyen}

\date{Universit\'e libre de Bruxelles (ULB), Facult\'e des Sciences, D\'epartement d'Informatique, Belgium \\
\texttt{gigeerae@ulb.ac.be\\ joel.goossens@ulb.ac.be\\ thi-van-anh.nguyen@ulb.ac.be}}

\maketitle

\begin{abstract}
  The online feasibility problem (for a set of sporadic tasks) asks
  whether there is a scheduler that always prevents deadline misses
  (if any), whatever the sequence of job releases, which is \textit{a
    priori} unknown to the scheduler. In the multiprocessor setting,
  this problem is notoriously difficult. The only exact test for this
  problem has been proposed by Bonifaci and Marchetti-Spaccamela: it
  consists in modelling all the possible behaviours of the scheduler
  and of the tasks as a graph; and to interpret this graph as a game
  between the tasks and the scheduler, which are seen as antagonistic
  players. Then, computing a correct scheduler is equivalent to
  finding a winning strategy for the `scheduler player', whose
  objective in the game is to avoid deadline misses. In practice,
  however this approach is limited by the intractable size of the
  graph. In this work, we consider the classical \emph{attractor}
  algorithm to solve such games, and introduce \emph{antichain
    techniques} to optimise its performance in practice and overcome
  the huge size of the game graph. These techniques are inspired from
  results from the formal methods community, and exploit the specific
  structure of the feasibility problem. We demonstrate empirically
  that our approach allows to dramatically improve the performance of
  the game solving algorithm.\end{abstract}

\section{Introduction}

The model of \emph{sporadic tasks} is nowadays a well-established
model for real-time systems. In this model, one considers a set of
real-time \emph{tasks}, that regularly release \emph{jobs}. The jobs
are the actual computational payloads, and the \emph{scheduler} must
assign those jobs to the available processor(s) of the platform in a
way that ensures no job misses its \emph{deadline} (i.e., ensuring
that all jobs are granted enough computation time within a given time
frame). The term \emph{sporadic} refers to the fact that there is some
uncertainty on the moment at which tasks release jobs: in the sporadic
model, an \emph{inter-arrival time} $T_i$ is associated with each task
$\tau_i$. This means that at least $T_i$ time units will elapse
in-between two job releases of the task $\tau_i$, but the release can
be delayed by an arbitrary long time (this is in contrast with the Liu
and Layland model~\cite{LL73} where the tasks release jobs
\emph{periodically}, i.e., every $T_i$ time unit).

In the present paper, we consider the so-called \emph{online
  feasibility} problem for such \emph{sporadic} tasks on a
\emph{multiprocessor platform}, i.e., a platform that boasts several
identical computational units on which the scheduler can choose to run
any active job. The \emph{online feasibility problem} asks, given a
set of (sporadic) tasks to \emph{compute}---if possible---a scheduling
policy for those tasks ensuring that no task ever misses its
deadline. An important assumption is that the scheduler is \emph{not
  clairvoyant}, i.e., it has no information on the future job
releases, apart from the information that can be deduced from the
current history of execution. Hence, the scheduler that has to be
computed must be able to react to any potential sequence of job
releases (it is thus an \emph{online} scheduler\footnote{Note that our
  approach requires a significant offline preprocessing phase, hence
  the schedulers that are produced can be characterised as
  \emph{semi}online, according to some authors.}).

This feasibility problem of sporadic tasks on multiprocessor platforms
is appealing both from the practical and the theoretical point of
view. On the practical side, the problem is rich enough to model
realistic situations. On the theoretical side, the different degrees
of non-determinism of the model (the exact release times of the jobs,
the exact execution duration of jobs, the many possible decisions of
the scheduler when there are more active tasks than available CPUs)
make its behaviour difficult to predict.

Indeed, while several
necessary and sufficient conditions have been identified for
feasibility, these criteria are not able to decide feasibility of all
systems. Moreover, while a collection of optimal schedulers for
\emph{uni}processors is known, \emph{optimal} online multiprocessor
scheduling of sporadic real-time tasks is known to be
impossible~\cite{FGB10}, at least for constrained or arbitrary
deadline systems.

In this paper, we build on the ideas introduced by Bonifaci and
Marchetti-Spaccamela \cite{DBLP:conf/esa/BonifaciM10}, and consider a
\emph{game-based model} for our feasibility problem. Intuitively, we
regard the scheduler as a player that competes against the coalition
of the tasks. The moves of the scheduler are the possible assignments
of active jobs to the CPUs, and the moves of the tasks are the
possible job releases. The objective of the scheduler in this game is
to avoid deadline misses, and we assume that the tasks have the
opposite objective. This models the fact that the tasks will not
collaborate (i.e., release jobs in a way that avoids deadline misses)
with the scheduler. Then, it should be clear that the system of tasks
is \emph{feasible} iff there is a \emph{winning strategy} for the
scheduler (one that enforces the scheduler's objective whatever the
tasks play). When the system is feasible, a scheduler can be extracted
from this winning strategy. To the best of our knowledge, this is the
only \emph{exact} feasibility test that has been proposed so far for
this problem.

Concretely, the game we consider is a special case of \emph{game
  played on graph}, that we call \emph{scheduling games}: the nodes of
the graph model the possible states of the system, and the edges model
the possible moves of the players. This approach is akin to the
\emph{controller synthesis problem} which has been extensively
considered in the formal methods community for the past 25 years
\cite{PR89}: given a model of a computer system interacting with an
environment and given a property that this system must satisfy
(typically, a safety property), the controller synthesis problem
consists in computing a \emph{controller} that enforces the property,
whatever the environment does. The parallel with our feasibility
problem should be clear: the scheduler corresponds to the controller,
the environment to the tasks, and the property is to avoid deadline
misses. In the formal methods literature, the classical approach to
the synthesis problem boils down to computing a winning strategy in a
game.

The main practical limitation to this approach is the intractable size
of the game graph---the so-called `state explosion problem', a typical
limitation of exhaustive verification (e.g. model-checking) and
synthesis techniques that consider all possible system sates. In our
case, this graph is exponential in the size of the problem, so,
building the whole graph, then analysing it (as proposed by Bonifaci
and Marchetti-Spaccamela \cite{DBLP:conf/esa/BonifaciM10}) is
infeasible in practice, even though the algorithms that compute
winning strategies for games played on graphs are efficient (they run
in polynomial time wrt the size of the graph). Nevertheless, the
community of formal methods has proposed several techniques (such as
efficient data structures, heuristics to avoid the exploration of all
states, \textit{etc}\ldots) that do overcome this barrier in practice
(see for example \cite{Bur92}). Building on these results, we have
shown, in previous works, that the same techniques have the potential
to dramatically improve the running time of some graph-based
algorithms solving real-time scheduling problems. More precisely, in
\cite{GGL-rts13} and \cite{TCS} we have applied so-called
\emph{antichain} \cite{DR10} techniques (introduced originally a.o. to
speed up model checking algorithms)  and demonstrated experimentally
their practical interest.

\noindent\textbf{Contributions} In this paper,
we continue our line of research and apply the antichain approach to
the improve the efficiency of the classical \emph{attractor} (or
\emph{backward induction}) algorithm to solve scheduling games. At the
core of our optimisation is a so-called \emph{simulation relation}
that allows one to avoid exploring a significant portion the game's
nodes. While the antichain approach is a generic one, the definition
of the simulation relation has to be tailored to the specific class of
problem under consideration. Our simulation relation exploits tightly
the peculiar structure of scheduling games. We perform a series of
experiments (see Section~\ref{sec:experimental-results}) to
demonstrate the usefulness of our new approach. Our enhanced algorithm
performs one order of magnitude better that the un-optimised attractor
algorithm (see \figurename~\ref{fig:es}). Moreover, we manage to
compute schedulers for systems that are not schedulable under the
classical EDF scheduler: in particular, when the load of the system is
close to 100\%, our approach outperforms vastly EDF (see
\figurename~\ref{fig:edf_bw}).

\noindent\textbf{Related works} The feasibility problem for sporadic tasks
is a well-studied problem in the real-time scheduling community.
Apart for the very particular case where for each task the deadline
corresponds to the period (implicit deadline systems), where
polynomial time feasibility tests exist, it exhibits a high
complexity. For instance, even in the uniprocessor case, the
feasibility problem of recurrent (periodic or sporadic) constrained
deadline tasks is strongly coNP complete~\cite{EkbergECRTS2015}. For
constrained/arbitrary deadline and sporadic task systems, several
\emph{necessary} or \emph{sufficient} conditions have been identified
for feasibility. However, these criteria are not sufficient to decide
feasibility for all systems: there are some systems that do not meet
any sufficient conditions (hence, are not surely feasible) and respect
all identified necessary conditions (hence are not surely
infeasible). We refer the reader to a survey by Davis and
Burns~\cite{davis2011survey}. 

As already mentioned, the only \emph{exact} test (i.e., algorithm to
decide the problem) is the one of Bonifaci and Marchetti-Spaccamela
\cite{DBLP:conf/esa/BonifaciM10} that consists in reducing the problem
to the computation of a winning strategy in a game. However no
optimisations are proposed to make this approach practical. In a
previous work \cite{TCS}, we have already presented a game solving
algorithm, enhanced by antichains, to solve the feasibility
problem. Compared to the present paper, \cite{TCS} relies on the same
model (Section~\ref{sec:feasibility-as-game}) and the same partial
order $\iesim$. However, the algorithm we consider here
(\algorithmcfname~\ref{alg:es}) is completely different. Its
improvement based on antichains (\algorithmcfname~\ref{alg:bw-tba},
discussed from Section~\ref{sec:improved-algorithm} onward) is
non-trivial and constitutes original work. Our experimental evaluation
(Section~\ref{sec:experimental-results}) compares our two approaches
and show that they are actually complementary.

\section{Feasibility as a game}\label{sec:feasibility-as-game}

Let us now define precisely the \emph{online feasibility problem} of
\emph{sporadic tasks} on a {\it multiprocessor platform}. Remember
that {\it online} means that the scheduler has no a priori knowledge
of the behaviour of the tasks; but must react, during the execution of
the systems to the requests made by the tasks. To formalise this
problem, we consider a set $\tau = \{\tau_1, \ldots, \tau_n\}$ of
sporadic tasks to be scheduled on $m$ identical processors. Each task
$\tau_i$ is characterised by three non-negative integer parameters
$C_i$, $D_i$, $T_i$. Parameter $C_i$ is the {\it worst case execution
  time} (WCET), i.e. an upper bound on the number of CPU ticks that
are needed to complete the job. $D_i$ is the \emph{deadline}, relative
to each job arrival. Each job needs to have been assigned $C_i$ units
of CPU time before its deadline, lest it will \emph{miss its deadline}
(for example if some job of task $\tau_i$ is released at instant $t$,
it will need $C_i$ units of CPU time before $t+D_i$). Finally, $T_i$
is the {\it minimal interarrival time}: at least $T_i$ time units must
elapse between two job releases of $\tau_i$. We assume that jobs can
be preempted and can be freely migrated between CPUs, without
incurring any delay. An example of system with three tasks is given in
\tablename~\ref{tab:running_ex}. Finally, we assume that each job
requires its WCET. Since the case where all jobs consume their WCET is
always a possible scenario, this assumption is without loss of
generality. Indeed, if there is no scheduler when the tasks are
restricted to consume their whole WCET, then there will be no
scheduler in the case where the tasks also have the option to complete
earlier; and, if jobs complete earlier at runtime, it is always
possible to keep the CPU idle for the remaining allocated ticks.

\begin{table}
\begin{center}
  \begin{tabular}{lccc}
    \toprule
             & $C_i$     & $D_i$    & $T_i$     \\ 
    \midrule
    $\tau_1$ & 1         & 1        & 2          \\ 
    $\tau_2$ & 2         & 2        & 2           \\ 
    $\tau_3$ & 1         & 4        & 2        \\ 
    \bottomrule
  \end{tabular}
 
\end{center}
\caption{The task set used as our running example}
\label{tab:running_ex}
\end{table}

In such a real-time system, tasks freely submit new jobs, respecting
their minimal interarrival time. Then, the scheduler is responsible to
allocate the jobs to the CPUs (ensuring that no job misses its
deadline), a CPU tick occurs (note that we consider a discrete time
model), and a new round begins with the tasks submitting new jobs,
etc. To model formally this semantics, we rely on the notion of
\emph{scheduling game}, as sketched in the introduction. Before
defining formally this notion, let us introduce it through an example.
Given a set of tasks $\tau$ and a number of CPUs $m$, a scheduling
game is played on a graph $\game_{\tau,m}$, by two players which are
the scheduler (called player \PS) and the coalition of the tasks
(player \PR). In the case of the task set $\tau$ of
\tablename~\ref{tab:running_ex} upon a 2-CPUs platform ($m=2$), a
prefix of $\game_{\tau,m}$ is given in
\figurename~\ref{fig:gameRS}. In this graph, the nodes model the
system states and the edges model the possible moves of the players
(all these notions will be defined precisely hereinafter). Each path is a
possible execution of the system. The actual path which is played
stems from the interaction of the players which play in turn: the
square (respectively round) nodes `belong' to \PR (\PS), which means
that \PR (\PS) decides in those nodes what will be the successor
node. The game starts in $\mathrm{init}^\PR$, where the tasks (\PR)
decide which jobs to submit. In the full graph $\game_{\tau,m}$, there
are $2^3=8$ different possibilities, but we display here only the case
where all tasks ($\tau_1$, $\tau_2$ and $\tau_3$) submit a job, which
moves the system to state $\PS_1$. In this state, the scheduler must
decide which tasks to schedule on the $m=2$ available CPUs. If it
decides to schedule only $\tau_1$, we obtain state $\PR_2$ after one
clock tick. From the game graph in \figurename~\ref{fig:gameRS}, it is
easy to see that this is bad choice for the scheduler. Indeed, from
$\PR_2$, player $\PR$ can move the system to $\PS_2$, and in all
successors ($\PR_4$,\ldots, $\PR_7$) of $\PS_2$ a task misses its
deadline (depicted by thick node borders). Hence, whatever the
scheduler plays from $\PS_2$ and $\PR_2$, he loses the game. On
the other hand, scheduling $\{\tau_1,\tau_2\}$ and $\{\tau_2,\tau_3\}$
from $\PS_1$ and $\PS_3$ respectively, guarantees the scheduler to win
the game (all successors of $\PR_{11}$ that are not shown are winning
too).

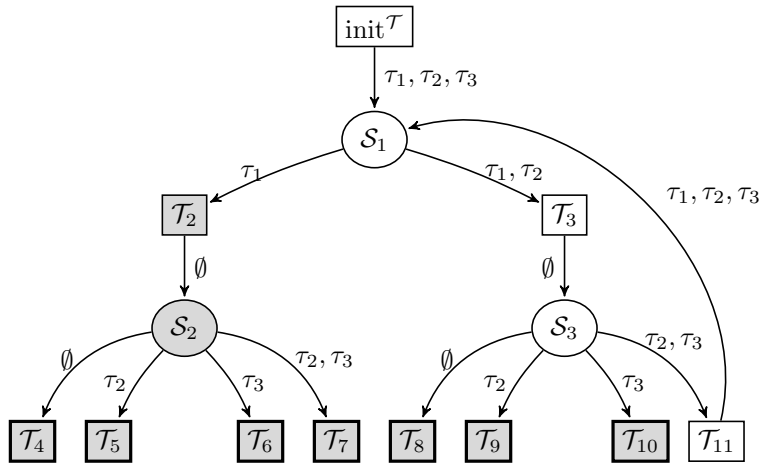
\begin{figure}
\begin{center}
\begin{tikzpicture}[->,>=stealth',shorten >=1pt,auto,node distance=1.5cm,
                    semithick]

  \tikzstyle{every state}=[text=black]

\node[rectangle,draw] (R1)  {init$^\PR$};
\node[ellipse, draw] (S1) [below of=R1]  { $\PS_1$};
\path (R1) edge [right]             node {$\tau_1, \tau_2,\tau_3$}  (S1);

\node[rectangle,draw,fill=gray!30] (R2) [below of=S1, xshift=-2.5cm, node distance=1cm]  {$\PR_2$};
\node[rectangle,draw] (R3) [below of=S1, xshift=2.5cm, node distance=1cm]  {$\PR_3$};

\path (S1) edge  [left, bend right=5]             node {$\tau_1$ } (R2)
         (S1) edge  [right, bend left=5]             node {$\tau_1, \tau_2$ }(R3);

\node[ellipse,draw,fill=gray!30] (S2) [below of=R2]  {$\PS_2$};
\path (R2) edge [right]             node {$\emptyset$ } (S2);

\node[rectangle,draw,fill=gray!30,very thick] (R4) [below of=S2, xshift=-2cm]  {$\PR_4$ };
\node[rectangle,draw,fill=gray!30,very thick] (R5) [below of=S2, xshift=-1cm]  {$\PR_5$ };
\node[rectangle,draw,fill=gray!30,very thick] (R6) [below of=S2, xshift=1cm]  {$\PR_6$ };
\node[rectangle,draw,fill=gray!30,very thick] (R7) [below of=S2, xshift=2cm]  {$\PR_7$};

\path (S2) edge [bend right,left]          node {$\emptyset$ } (R4)
	(S2) edge [bend right=10, left] node {$\tau_2$} (R5)
	(S2) edge [bend left=10, right] node{$\tau_3$} (R6)
         (S2) edge [bend left,right] node{$\tau_2, \tau_3$} (R7);

\node[ellipse,draw] (S3) [below of=R3]  {$\PS_3$};
\path (R3) edge [left]             node {$\emptyset$ } (S3);

\node[rectangle,draw,fill=gray!30,very thick] (R8) [below of=S3, xshift=-2cm]  {$\PR_8$};
\node[rectangle,draw,fill=gray!30,very thick] (R9) [below of=S3, xshift=-1cm]  {$\PR_9$};
\node[rectangle,draw,fill=gray!30,very thick] (R10) [below of=S3, xshift=1cm]  {$\PR_{10}$};
\node[rectangle,draw] (R11) [below of=S3, xshift=2cm]  {$\PR_{11}$};

\path (S3) edge [bend right,left]          node {$\emptyset$ } (R8)
	(S3) edge [bend right=10, left] node {$\tau_2$} (R9)
	(S3) edge [bend left=10, right] node{$\tau_3$} (R10)
         (S3) edge [bend left,right] node[near start]{$\tau_2, \tau_3$} (R11);

\path (R11) edge [bend right=60]             node[right of = R3, xshift=-0.5cm] {$\tau_1, \tau_2, \tau_3$ } (S1);

\end{tikzpicture}

\caption{A prefix of the scheduling game corresponding to the system
  in \tablename~\ref{tab:running_ex}. Square nodes belong to \PR and
  the round ones to \PS. Gray nodes are \emph{losing}. Nodes with thick
  borders are \emph{bad} nodes, i.e. where at least one task misses
  its deadline.}
\label{fig:gameRS}
\end{center}

\end{figure}

\subsection{Scheduling games}

Let us now formally define scheduling games. We start by the notion of
system state that models the current status of all tasks in the system
\cite{BC-opodis07}. In all states $S$ of the system, we store two
pieces of information for each task $\tau_i$:
\begin{inparaenum}[(i)]
\item The earliest next arrival time $\nat_S(\tau_i)$ of the next job of
  $\tau_i$ and
\item the remaining computing time $\rct_S(\tau_i)$ of the current
  job\footnote{If several jobs of the same task are active at the same
  time, they are treated one at a time, in FIFO order.}
  of $\tau_i$.
\end{inparaenum}

\begin{definition}[System states]
  Let $\tau = \{\tau_1, \ldots, \tau_n\}$ be a set of sporadic
  tasks. A system state $S$ of $\tau$ is a pair ($\nat_S, \rct_S$)
  where:
\begin{inparaenum}[(i)]
\item $\nat_S$ is a function assigning to all tasks $\tau_i$ a value
  $\nat_S(\tau_i)\leq T_i$; and
\item $\rct_S$ is a function assigning to all tasks $\tau_i$ a value
  $\rct_S(\tau_i)\leq C_i$
\end{inparaenum}
\end{definition}
Observe that, since the $\nat$ parameter is not bounded below, there
are infinitely many system states. We will limit the number of states
to a finite set when defining the game.  A task $\tau_i$ is said to be
{\em active} in state $S$ if it currently has a job that has not
finished in $S$, i.e., the set of active tasks in $S$ is
$\act(S) = \{ \tau_i \mid \rct_S(\tau_i ) > 0 \}$. To the contrary, we
say that $\tau_i$ is \emph{idle} in $S$ iff $\rct_S(\tau_i ) = 0$.  A
task $\tau_i$ is {\em eligible} in $S$ if it can submit a job from
this configuration, i.e., the set of eligible tasks in $S$ is:
$\eligible(S) = \{ \tau_i \mid \nat_S(\tau_i) \leq 0 \wedge
\rct_S(\tau_i) = 0 \}$.
Finally, The \emph{laxity} of $\tau_i$ in $S$ is the value
$\laxity_S(\tau_i) = \nat_S(\tau_i) - (T_i - D_i) - \rct_S(\tau_i)$.
Intuitively, the laxity of a task measures the amount of forthcoming
CPU steps that the task could remain idle without taking the risk to
miss its deadline. In particular, states where the laxity is negative
will be declared as losing states.

Now let us define the possible moves of both players. Let
$S\in\states$ be a system state. We first define the possible moves of
player \PS, i.e.\ the scheduler. One move of \PS is characterised by
the choice of the set $\tau'$ of tasks that will be
scheduled. Formally, for all $\tau'\subseteq \act(S)$ s.t.\
$|\tau'|\leq m$ (i.e., $\tau'$ does not contain more tasks than the
$m$ available CPUs), we let $\succS(S,\tau')$ be the (unique) state
$S'$ s.t.:
\begin{inparaenum}[(i)]
\item $\nat_{S'}(\tau_i) = \nat_{S} (\tau_i) -1$; and
\item $\rct_{S'}(\tau_i)=\rct_{S}(\tau_i)$ if $\tau_i \not\in
  \tau'$. Otherwise, $\rct_{S'}(\tau_i)=\rct_{S}(\tau_i) -1$.
\end{inparaenum}
Let us now define the possible moves of player \PR, i.e.\ the
tasks. Each move of \PR is characterised by a set $\tau'$ of eligible
tasks. We let $\succR(S,\tau')\subseteq\states$ be the set of system
states s.t.\ $S'\in\succR(S,\tau')$ iff:
\begin{inparaenum}[(i)]
\item $\nat_{S'}(\tau_i) = \nat_{S} (\tau_i)$ if
  $\tau_i \not\in \tau' $. Otherwise,
  $\nat_{S'}(\tau_i) = T_i$;
\item And if $\tau_i \not\in
  \tau'$, $\rct_{S'}(\tau_i) =\rct_{S}(\tau_i)$. Otherwise, $\rct_{S'}(\tau_i) = C_i$.
\end{inparaenum}
Finally, we let $\Succ{S}=\succS\cup\succR$.

\begin{example}
  The initial state of the game $\game_{\tau,m}$ in our running
  example (state $\mathrm{init}^\PR$ in \figurename~\ref{fig:gameRS})
  contains the system state $\tuple{(0,0),(0,0),(0,0)}$ (we denote a
  system state $S$ by the tuple
  $\tuple{(\rct(\tau_1),\nat(\tau_1)),\ldots,
    (\rct(\tau_n),\nat(\tau_n))}$).  Then,
  $\PS_1=\tuple{(1,2), (2,2), (1,2)}$. In $\PS_1$ all tasks are
  eligible, $\laxity_{\PS_1}(\tau_1)=0$ and
  $\laxity_{\PS_1}(\tau_3)=3$. Also,
  $\PR_3=\tuple{(0,1),(1,1),(1,1)}$, which shows that the job of
  $\tau_1$ that had been initially submitted has now completed (so,
  $\tau_1$ is idle in $\PR_3$), but we must still wait one time unit
  before all tasks can submit a new job (their $\nat$ are all equal to
  $1$). This explains why the only successor of $\PR_3$ is
  $\succR(\PR_3,\emptyset)$. Finally, $\tau_2$ has now missed its
  deadline in $\PR_{10}=\tuple{(0,0),(1,0),(0,0)}$. Indeed
  $\laxity_{\PR_{10}}(\tau_2)<0$.
  \end{example}

  For a set $\tau = \{\tau_1, \ldots, \tau_n \}$ of sporadic tasks,
  and $m$ CPUs, we let
  $\game_{\tau,m}=\tuple{V_{\PS}, V_{\PR}, E, I, \bad}$ where:
\begin{itemize}
\item
  $V_{\PS}= \{S\in\states\mid \forall 1\leq i\leq n:
  \laxity_S(\tau_i)\geq -1\} \times \{\PS\}$
  is the set of scheduler-controlled nodes.
\item
  $V_{\PR} = \{S\in\states\mid \forall 1\leq i\leq n:
  \laxity_S(\tau_i)\geq -1\}\times \{\PR\}$
  is the set of nodes controlled by the tasks.
\item $E = E_{\PS} \cup E_{\PR} $ is the set of edges where:
  \begin{itemize}
  \item $E_{\PS}$ is the set of scheduler moves. It contains an edge
    $\big((S,\PS), (S',\PR)\big)$ iff there is
    $\tau' \subseteq \act(S)$ s.t.\  $|\tau'| \leq m$ and
    $S' = \succS(S, \tau')$. In this case, we sometimes abuse the
    notations, and consider that this edge is labelled by $\tau'$,
    denoting it $\big((S,\PS), \tau', (S',\PR)\big)$.
  \item $E_{\PR}$ is the set of tasks moves. It contains an edge
    $\big( (S,\PR), (S',\PS)\big)$ iff there exists
    $\tau' \subseteq \eligible(S)$ s.t.\ $S' \in \succR(S,
    \tau')$. Again, we abuse notations and denote this edge by $\big(
    (S,\PR), \tau', (S',\PS)\big)$.
  \end{itemize}
\item $I = (S_0,\PR)$, where for all $1\leq i\leq n$:
  $\rct_{S_0}(\tau_i) = \nat_{S_0}(\tau_i) = 0$ is the initial
  state.
\item
  $\bad = \{(S,\PR) \in V_{\PR} \mid \exists \tau_i \in \act(S)$
    such that $\laxity_S(\tau_i) < 0\}$,
  i.e.\ $\bad$ is the set of failure states.
\end{itemize}
Observe that in this definition, we consider only states where the
laxity of all tasks is $\geq -1$. It is easy to check (see
\cite{GGL-rts13,BC-opodis07} for the details) that this restriction
guarantees the number of game nodes to be \emph{finite}. Moreover,
this is sufficient to detect all deadline misses, as any execution of
the system where a task misses a deadline will necessarily traverse a
state with a laxity equal to $-1$ for one of the tasks
\cite{GGL-rts13,BC-opodis07}. In the sequel we lift the $\nat$ and
$\rct$ notations to nodes of the games, i.e., for a node
$v=(S,\mathcal{P})$, we let $\nat_v=\nat_S$ and $\rct_v=\rct_S$.

\subsection{Computation of winning strategies}
Now note that the syntax of scheduling games is fixed, let us explain what
we want to compute on those games, i.e. \emph{winning strategies} for
player \PS. Fix a scheduling game
$\game_{\tau,m}=\tuple{V_{\PS}, V_{\PR}, E, I, \bad}$. First, a
\emph{play} in $\game_{\tau,m}$ is a path in the game graph, i.e. a
sequence $v_0,v_1,\ldots, v_i,\ldots$ of graph vertices (for all
$i\geq 0$: $v_i\in V_{\PS}\cup V_{\PR}$) s.t. for all $i\geq 0$:
$(v_i,v_{i+1})\in E$.  Then, a \emph{strategy} for player $\PS$ is a
function $\sigma:V_{\PS}\rightarrow V_{\PR}$ s.t. for all
$v\in V_{\PS}$: $(v,\sigma(v))\in E$. That is, $\sigma$ assigns, to
each node of $V_{\PS}$, one of its successor
\footnote{In general,
  strategies might need to have access to the whole prefix of
  execution to determine what to play next. However, since scheduling
  games are a special case of safety games \cite{Tho95}, it is
  well-known that \emph{positional} strategies are sufficient,
  i.e. strategies whose output depends only on the current game
  vertex.} 
to be played when $v$ is reached. Given a strategy
$\sigma$, the \emph{outcome} of $\sigma$ from a node
$v\in V_\PS\cup V_\PR$ is the set $\outcome{\sigma,v}$ of all possible
plays in $\game_{\tau,m}$ that start in $v$ and are obtained with
$\PS$ playing according to $\sigma$, i.e., 
$v_0,v_1,\ldots, v_i,\ldots \in \outcome{\sigma,v}$ iff
\begin{inparaenum}[(i)]
\item $v_0=v$; and
\item for all $i\geq 0$: $v_i\in V_{\PS}$ implies
  $v_{i+1}=\sigma(v_i)$.
\end{inparaenum}
We denote by $\outcome{\sigma}$ the set $\outcome{\sigma, I}$.  Then,
a strategy $\sigma$ is \emph{winning} from a node $v$ iff it ensures
$\PS$ to always avoid $\bad$ (from $v$) whatever $\PR$ does, i.e.,
$\sigma$ is winning from $v$ iff for all plays
$v_0,v_1,\ldots,v_i,\ldots\in\outcome{\sigma,v}$, for all $i\geq 0$:
$v_i\not\in \bad$. A strategy $\sigma$ is simply said \emph{winning}
if it is winning from the initial node $I$. We say that a node $v$ is
\emph{winning} (respectively \emph{losing}) iff there exists a (there
is no) strategy that is winning from $v$. We denote by $\win$ and
$\losing$ the sets of winning and losing nodes respectively. It is
well-known \footnote{This property is called \emph{determinacy} and
  can be established using the classical result of Donald Martin
  \cite{10.2307/1971035}.}  that, in our setting, all nodes are either
winning or losing, i.e., $\win\cup\losing=V_\PR\cup V_\PS$.  Hence,
computing one of those sets is sufficient to deduce the other.

It is easy to check that, given a set of sporadic tasks $\tau$ and a
number $m$ of CPUs, the scheduling game $\game_{\tau,m}$ captures all
the possible behaviours of $\tau$ on a platform of $m$ identical CPUs,
running under any possible online scheduler. Then, the answer to the
online feasibility problem is positive (i.e., there exists an online
scheduler) iff there is a winning strategy in the game, i.e., iff
$I\in\win$, or, equivalently, iff $I\not\in\losing$ (this is the
condition that our algorithms will check). Actually, the winning
strategy $\sigma$ is the scheduler itself: by definition, every pair
$(v,\sigma(v))$ corresponds to an edge $(v,\tau',\sigma(v))$ in
$\game_{\tau,m'}$, which means that, in state $v$, the scheduler must
grant the CPUs to all tasks in $\tau'$. Since $\sigma$ is winning no
bad state will ever be reached, hence, no deadline will be missed.
\begin{example}
  In the example graph of \figurename~\ref{fig:gameRS}, a winning
  strategy is a strategy $\sigma$ s.t. $\sigma(\PS_1)=\PR_3$ and
  $\sigma(\PS_3)=\PR_{11}$. It is easy to check that (on this graph
  excerpt), playing this strategy allows to avoid
  $\bad=\{\PR_4,\PR_5,\ldots,\PR_{10}\}$, whatever choice \PR makes.
\end{example}

Let us now recall a classical algorithm to solve \emph{safety games},
a broad class of games played on graphs to which scheduling games
belong. The algorithm is a variant of \emph{backward induction}: it
consists in computing \emph{the set of all nodes from which \PS cannot
  guarantee to avoid \bad}. This set, denoted $\Attr$ is called the
\emph{attractor} of $\bad$; in \figurename~\ref{fig:gameRS}, it
corresponds to the gray nodes. The computation is done inductively,
starting from \bad, and following the graph edges in a backward
fashion (hence the name). Formally, the algorithm relies on the two
following operators, for a set $V\subseteq V_{\PR}\cup V_{\PS}$.
\begin{eqnarray}
 \label{eq:pree1}\label{eq:pree2}
 \PreE V &=& \{v \mid \Succ v \cap V \neq \emptyset \} \\
 \label{eq:prea1}
 \PreA V &=& \{v \mid \Succ v \subseteq V\}.
\end{eqnarray}

That is, $v\in \PreE{V}$ iff \emph{there exists a successor} of $v$
which is in $V$; and $v\in\PreA{V}$ iff \emph{all successors} of $v$
belong to $V$.  Then, given $V\subseteq V_{\PR}\cup V_{\PS}$, we let
$\UPre{V}$ be the set of \emph{uncontrollable predecessors} of $V$:
$\UPre{V} = \PreE{V \cap V_\PS} \cup \PreA{V \cap V_\PR}$.
Intuitively, $v\in\UPre{V}$ iff $\PS$ cannot prevent the game to reach
$V$ (in one step) from $v$, either because $v$ is a $\PR$-node that
has a successor (which is necessarily a $\PS$-node) in $V$, or because
$v$ is an $\PS$-node that has all its successors (that are all
$\PR$-nodes) in $V$. Thus, if $V$ contains only losing nodes, then all
nodes in $\UPre{V}$ are losing too.

\begin{algorithm}
  \DontPrintSemicolon
  \ES
  \Begin{
    $i\leftarrow 0$\;
    $\Attr_0\leftarrow\bad$ \;
    \Repeat {$\Attr_i=\Attr_{i-1}$}{
      $\Attr_{i+1}\leftarrow \Attr_i\cup \UPre{\Attr_i}$ \;
      $i\leftarrow i+1$ \;
    }
    \Return $\Attr_i$ \;
  }
  \caption{Backward algorithm to compute $\losing$.}
  \label{alg:es}
\end{algorithm}

Equipped with these definitions, we can describe now the backward
algorithm (\algorithmcfname~\ref{alg:es}) for solving scheduling
games, that will be the basis of our contribution. It computes a
sequence $\left(\Attr_i\right)_{i\geq 0}$ of sets of states with the
following property: $v\in\Attr_i$ iff player $\PR$ has a strategy to
force the game to reach $\bad$ in at most $i$ steps from $v$. Since
the sets $\Attr_i$ (which contain nodes of the game graph) grow along
with the iterations of the algorithm, and since the graph is finite,
the algorithm necessarily terminates. It is well-known that it
computes exactly the set of losing states:
\begin{theorem}[Taken from \cite{Tho95}]\label{the:es-is-correct}
  When applied to a \\ scheduling game
  $\game=\tuple{V_{\PS}, V_{\PR}, E, I, \bad}$,
  \algorithmcfname~\ref{alg:es} terminates in at most $|V_{\PR}\cup
  V_{\PS}|$ steps and returns the set of losing states $\losing$.
\end{theorem}
Thus, the answer to the online feasibility problem is positive iff $I$
does not belong to the set returned by the algorithm. When it is the
case, a winning strategy (hence also a scheduler) can easily be
obtained: in all nodes $v\in V_{\PS}$ that are visited, we let
$\sigma(v)$ be any node $v'$ s.t. $(v,v')\in E$ and
$v'\not\in\losing$. Such a node $v'$ is guaranteed to exist by
construction (otherwise, by definition of $\forall\Prename$, $v$ would
have been added to the attractor, and would not be winning).
\begin{example}
  On the graph in \figurename~\ref{fig:gameRS}, we have:
  $\Attr_0=\{\PR_4,\PR_5,\ldots,\PR_{10}\}$;
  $\Attr_1=\Attr_0\cup\{\PS_2\}$; $\Attr_2=\Attr_1\cup\{\PR_2\}$; and
  $\Attr_2=\Attr_3$, so the algorithm converges after 3 steps.
\end{example}

While this algorithm is a nice theoretical solution, one needs to
answer the following questions to obtain a practical implementation:
\begin{inparaenum}[(i)]
\item how can we compute a representation of $\bad$? and
\item given $\Attr_i$, how can we compute $\Attr_{i+1}$?
\end{inparaenum}
Since the game graph is finite, these two questions can be solved by
building the whole game graph, and analysing it. However, in practice,
the graph has a size which is exponential in the description of the
problem instance \cite{DBLP:conf/esa/BonifaciM10}, so applying this
algorithm straightforwardly is not possible in practice. The core
contribution of the present work is to propose heuristics that exploit
the particular structure of scheduling games, in order to speed up
\algorithmcfname~\ref{alg:es}.

\section{Antichain techniques for  sporadic tasks\label{sec:antich-techn-spor}}

In this section, we introduce the main technical contribution of the
paper: an \emph{antichain}-based \cite{DR10} heuristic to speed up the
performance of the backward algorithm for solving scheduling
games. Let us start with an intuitive discussion of our heuristic. It
relies on the definition of a partial order $\iesim$ that compares the
state of the game graph, with the following important property:
\emph{if $v$ is a losing node, then, all `bigger' nodes $v'$ (i.e.,
  with $v'\iesim v$) are also losing}\footnote{Actually, the
  symmetrical property is also true: if $v\in\win$, then, $v\iesim v'$
  implies $v'\in\win$.}. Then, the optimisation of
\algorithmcfname~\ref{alg:es} consists, roughly speaking, in
\emph{keeping in the sets $\Attr_i$ the minimal elements only}. This
has two consequences. First, since there are, in practice, many
comparable elements in the sets $\Attr_i$, the memory footprint of the
algorithm is dramatically reduced. Second, we perform the computation
of the uncontrollable predecessors on the minimal elements only
(instead of computing the uncontrollable predecessors of all the
elements in $\Attr_i$), which reduces significantly the running time
of each iteration of the algorithm. It should now be clear that the
definition of $\iesim$ partial order is central to our improved
algorithm. Its formal definition will be given later, but we can
already sketch its intuition through the following example:
\begin{example}
  On the graph given in \figurename~\ref{fig:gameRS}, we have
  $\PR_4\iesim \PR_5$. This means that since $\PR_5$ is losing (it is
  in $\bad$), then $\PR_4$ should be too (it is indeed the case). So,
  intuitively, $\PR_4$ is a node which (from the scheduler's point of
  view) is `harder' to win from than $\PR_5$. Indeed, one can check
  that $\PR_4$ corresponds to system state $\tuple{(0,1),(2,0),(1,0)}$
  and $\PR_5$ to $\tuple{(0,1),(1,0),(1,0)}$. So,
  $\rct_{\PR_4}(\tau_2)>\rct_{\PR_5}(\tau_2)$, while all the other
  parameters of the states are equal. It is thus not surprising that
  $\PR_4$ is `harder' than $\PR_5$, since $\tau_2$ needs more CPU time
  in $\PR_4$ than in $\PR_5$ with the same time remaining to its
  deadline. This also means that in our improved algorithm, we will
  not keep $\PR_4$ in the initial set of states we compute $\Attr_0$,
  since $\PR_4$ is not minimal in $\bad$.
\end{example}

\subsection{Partial orders, antichains and closed sets}
Let us now introduce the definitions on which our techniques rely.
Fix a finite set $S$. A relation $\wbs\in S\times S$ is a partial
order iff $\wbs$ is reflexive, transitive and antisymmetric.  As
usual, we often write $s\wbs s'$ and $s\nwbs s'$ instead of
$(s,s')\in \wbs$ and $(s,s')\not\in \wbs$, respectively. The
$\wbs$-\emph{upward closure} $\uco{S'}{\wbs}$ of a set $S'\subseteq S$
is the smallest set containing $S'$ and all the elements that are
larger than some element in $S$', i.e.,
$\uco{S'}{\wbs}=\{s \mid \exists s'\in S': s\wbs s'\}$.  Then, a set
$S'$ is \emph{upward closed} iff $S' = \uco{S'}{\wbs}$.  When the
partial order is clear from the context, we often write $\uc{S}$
instead of $\uco{S}{\wbs}$. A subset $\ATC$ of some set
$S'\subseteq S$ is an \emph{antichain} on $S'$ with respect to $\wbs$
iff it contains only elements that are incomparable wrt $\wbs$,
i.e. for all $s, s' \in \ATC$: $s\neq s'$ implies $s\nwbs s'$.  An
antichain $\ATC$ on $S'$ is said to be a set of \emph{minimal
  elements of $S'$} (or \emph{a minimal antichain} of $S'$) iff for
all $s_1\in S'$ there is $s_2 \in \ATC$: $s_1\wbs s_2$. It is easy
to check that if $\ATC$ is a minimal antichain of $S'$, then
$\uc{\ATC}=\uc{S'}$. This is a key observation for our improved
algorithm: intuitively, $\ATC$ can be regarded as a \emph{compact}
representation of $\uc{S'}$, which is of minimal size in the sense
that it contains no pair of $\wbs$-comparable elements. Moreover,
since $\wbs$ is a partial order, each subset $S'$ of the finite set
$S$ admits a unique minimal antichain, that we denote by
$\minac{S'}$. Observe that one can always effectively build
$\minac{S'}$, simply by iteratively removing from $S'$, all the
elements that strictly dominate another one.

In our game setting, we will be interested in a particular class of
partial orders, which are called \emph{turn-based alternating
  simulations}:
\begin{definition}[\cite{AHKV-concur98,DBLP:conf/rp/GeeraertsGS14}]\label{def:tba}
  Let $G=(V_\PS,V_\PR,E,I,\bad)$ be a safety game. A partial
  order $\wbs\subseteq (V_\PS\times V_\PS) \cup (V_\PR\times V_\PR)$ is a
  \emph{turn-based alternating simulation relation for $G$}
  (tba-simulation for short) iff for all $v_1$,
  $v_2$ s.t.\ $v_1\wbs v_2$, either $v_1\in \bad$ or the three
  following conditions hold:
  \begin{inparaenum}[(i)]
  \item If $v_1\in V_\PS$, then, for all $v_1'\in\Succ{v_1}$, there is
    $v_2'\in\Succ{v_2}$ s.t.\  $v'_1\wbs v'_2$;
  \item If $v_1\in V_\PR$, then, for all $v_2'\in\Succ{v_2}$, there is
    $v'_1\in\Succ{v_1}$ s.t.\ $v'_1\wbs v'_2$; and
  \item $v_2 \in \bad$ implies $v_1 \in \bad$.
  \end{inparaenum}
\end{definition}
\todo[inline]{Some intuition? G.}

\subsection{Improved algorithm}\label{sec:improved-algorithm}
The key observation to our improved algorithm can now be stated
formally: the sets $\Attr_i$ computed by \algorithmcfname~\ref{alg:es}
are actually \emph{upward-closed} sets (for any partial order which is
a tba-simulation), which justifies that we can a manipulate them
compactly through their minimal elements only:

\begin{lemma} \label{lem:upre} Given a set $V$ which is upward-closed
  for a tba-simulation $\wbs$: $\UPre{ V} = \uco{ \UPre{
      V}}{\wbs}$. \end{lemma}
\begin{proof}
   We need to prove that $\UPre{ V} \subseteq \uc{\UPre{ V}}$ and
  $\uc{\UPre{V}} \subseteq\UPre{V}$. The first item is trivially
  correct. Now we will prove the second item by showing that
  $\forall v \in \uc{\UPre{V}} $ then $v \in \UPre{V}$. Since
  $v \in \uc{\UPre{ V}}$, there hence exists $v' \in \UPre{V}$
  such that $v \wbs v'$. We divide into two cases as follows.
\begin{itemize}
\item In case $v, v' \in V \cap V_\PR$: By Definition~\ref{def:tba},
  for all $\overline{v'} \in \Succ{v'}$, there exists
  $\overline{v} \in \Succ{v}$ such that
  $\overline{v} \wbs \overline{v'}$. By the definition of
  $\PreA{V \cap V_\PR}$ (see equation \eqref{eq:prea1}), $\overline{v'} \in V$
  therefore $\overline{v} \in V$. Finally, by the definition of
  $\PreE{V \cap V_\PS}$ (see \eqref{eq:pree1}), we derive that
  $v \in \PreE{V \cap V_\PS}$ therefore $v \in \UPre{V}$.
\item In case $v, v' \in V \cap V_\PS$: By Definition \ref{def:tba},
  for all $\overline{v} \in \Succ{v}$, then there exists
  $\overline{v'} \in \Succ{v'}$ such that
  $\overline{v} \wbs \overline{v'}$. Then since $v' \in \UPre{V}$ then
  $\Succ{v'} \subseteq V$ (by the definition of
  $\PreA{V \cup V_\PR}$).  Therefore $\Succ{v} \subseteq V$, we derive
  that $v \in \PreA{V \cup V_\PR}$ (see equation \eqref{eq:prea1}) and
  finally $v \in \UPre{V}$.
\end{itemize}
\end{proof}
    
\begin{proposition} \label{prop:attr} Let $\game$ be a scheduling game
  and let $\wbs$ a tba-simulation for $\game$. Then the sets $\Attr_i$
  computed are upward-closed for $\wbs$, i.e. for all $i\geq 0$:
  $\Attr_i = \uco{\Attr_i}{\wbs}$.
\end{proposition}
\begin{proof}
  We prove the lemma by induction on $i$.

{\bf Base case: } When $i = 0$, by the definition, $\Attr_0 =
\bad$.
We need to prove that $\bad = \uc \bad$. It is trivial to
conclude that $\bad \subseteq \uc \bad$. It remains to prove that
$\uc \bad \subseteq \bad$. For each $v \in \uc \bad$, there
exists $v' \in \bad$ such that $v \wbs v'$. Since $v'$ is a losing
state, than $v'$ is losing as well (by the definition of
tba-simulation). Hence, $v \in \bad$.

{\bf Inductive case: } For all $i \geq 1$, given
$\Attr_i = \uc{\Attr_{i}}$, we will prove that
$\Attr_{i+1} = \uc{\Attr_{i+1}}$. By the definition of the
sequences of $\Attr$, $ \Attr_{i+1} = \Attr_i \cup
\UPre{\Attr_i}$.
Then we need to prove that:

\[
\Attr_i \cup \UPre{\Attr_i} = \uc{\Attr_i \cup
\UPre{\Attr_i}},
\]
i.e. $\UPre{\Attr_i} = \uc{\UPre{\Attr_i}}$ that has been done by
Lemma~\ref{lem:upre}.
\end{proof}

We can now define the improvement to \algorithmcfname~\ref{alg:es},
that consists in manipulating the sequence of sets
$(\Attr_i)_{i\geq 0}$ by their minimal elements only. To this end, we
define $\UPreATCname$, a variant of the $\UPrename$ operator that
works directly on the minimal elements. It receives an antichain of
minimal elements $\ATC$ (that represents the upward-closed set
$\uc{\ATC}$) and returns the antichain of minimal elements of
$\UPre{\uc{\ATC}}$:
\begin{align}
\label{eq:prea-atc}
\UPreATC{\ATC} &= \minac{\UPre{\uc{\ATC}}}
\end{align}
Clearly, since the game graph is finite and since $\UPrename$ can be
computed, $\PreAATC{\ATC}$ is computable too: it suffices to `expand'
the antichain $\ATC$ to $\uc{\ATC}$ (which is a finite set), compute
$\UPre{\uc{\ATC}}$, then keep only the minimal elements. However, this
procedure deceives the purpose of using a compact representation of
upward-closed sets, so we will introduce later an efficient way to
compute $\PreAATC{\ATC}$ for a given partial order. Equipped with
$\UPreATCname$, we can now define our improved algorithm, given in
\algorithmcfname~\ref{alg:bw-tba}. As announced, it computes a
sequence $(\ALS_i)_{i\geq 0}$ of antichains representing respectively
the sets in the $(\Attr_i)_{i\geq 0}$ sequence, as stated by the
following proposition:

\begin{algorithm}
  \DontPrintSemicolon
  \Bw
  \Begin{
    $i\leftarrow 0$\;
    $\ALS_0\leftarrow \lfloor \bad \rfloor$ \;
    \Repeat {$\ALS_{i}=\ALS_{i-1}$}{
      $\ALS_{i+1}\leftarrow \lfloor\ALS_i\cup \UPreATC{\ALS_i}\rfloor$ \;
      $i\leftarrow i+1$ \;
    }
    \Return $\ALS_i$ \;
  }
  \caption{Backward traversal with tba-simulation, for the computation
  of $\minac{\losing}$}
  \label{alg:bw-tba}

  \end{algorithm}

\begin{proposition}
\label{pro:atr-als}
For all $i \in \mathbb{N}$: $\Attr_i = \uc{\ALS_i}$
\end{proposition}
\begin{proof}
  The proof is by induction on $i$.

\textbf{Base case ($i=0$)}: When $i=0$, $\Attr_0 = \bad$ and
$\ALS_0 = \minac{\bad}$. This case is hence trivially correct.

\textbf{Inductive case ($i=k$)}:  The induction hypothesis is that,
$\Attr_{k-1} = \uc{\ALS_{k-1}}$. Let us prove that
$\Attr_{k} = \uparrow \ALS_{k}$. By the definitions of the two
sequences, it is necessary to prove that 
\begin{align*}
  & \Attr_{k-1} \cup \UPre{\Attr_{k-1}}\\
&= \uc{\ALS_{k-1} \cup \UPreATC
  {\ALS_{k-1}}}\\
&= \uc{\ALS_{k-1}}\cup\uc{\UPreATC
  {\ALS_{k-1}}}.
\end{align*}
However:
 \begin{align}
      \Attr_{k-1} &= \uparrow \ALS_{k-1}
 \end{align}
 by induction hypothesis and:
 \begin{align}
     \UPre{\Attr_i} &=  \uc{\UPreATC {\ALS_i}}
 \end{align}
 by definition of $\UPreATCname$, see equation~\eqref{eq:prea-atc}.
\end{proof}

Based on this proposition, and on the fact that
\algorithmcfname~\ref{alg:es} computes the sequence
$(\Attr_i)_{i\geq 0}$ and returns $\losing$, we conclude that
\algorithmcfname~\ref{alg:bw-tba} is correct, in the sense that it
computes the minimal elements of $\losing$:
\begin{theorem} \label{theo:termination}
  When applied to a scheduling game $\game$ equip\-ped with a
  tba-simulation $\wbs$, \algorithmcfname~\ref{alg:bw-tba} terminates
  and returns $\minac{\losing}$.
\end{theorem}
\begin{proof}
  By Theorem~\ref{the:es-is-correct}, we know that there is $k$
  s.t. $\Attr_k=\Attr_{k+1}=\losing$. Hence, by
  Proposition~\ref{pro:atr-als},
  $\ALS_k=\minac{\Attr_k}=\minac{\Attr_{k+1}}=\ALS_{k+1}$. Hence,
  \algorithmcfname~\ref{alg:bw-tba} terminates, and returns a set which
  is equal to $\minac{\Attr_k}=\minac{\losing}$ (remark that
  \algorithmcfname~\ref{alg:bw-tba} could converge earlier than
  \algorithmcfname~\ref{alg:es}).
\end{proof}
\section{The improved algorithm in practice}
 \label{sec:improved-algo}

 Until now, we have proven that given a scheduling game $\game$, and a
 tba-simulation $\wbs$ for $\game$, Algorithm~\ref{alg:bw-tba}
 terminates and returns the minimal elements of the set of losing
 states. However, this algorithm is parameterised by the definition of
 a proper tba-simulation (from which the definition of $\UPreATCname$
 depends too), which is strong enough to ensure that the sets
 $(\ALS_i)_{i\geq 0}$ will be much smaller than their
 $(\Attr_i)_{i\geq 0}$ counterparts. The purpose of this section is
 twofold:
\begin{inparaenum}[(i)]
\item to define a tba-simulation $\iesim$ that exploits adequately the
  structure of scheduling games, as sketched above; and
\item to show how to compute efficiently the set $\minac{\bad}$ and the
  operator $\UPreATCname$ (based on $\iesim$).
\end{inparaenum}

For the tba-simulation, we rely on our previous works \cite{TCS},
where we have introduced the \emph{idle-ext task simulation} partial
order and proved that it is tba-simulation:

\begin{definition}[Idle-ext task simulation] \label{def:idle} Let
  $\tau$ be a set of sporadic tasks on a platform of $m$ processors
  and let $G = (V_{\PS}, V_{\PR}, E, I, \bad)$ be a scheduling game of
  $\tau$ on the platform. Then, the {\em idle-ext tasks partial order}
  $\iesim \subseteq V_{\PS} \times V_{\PS} \cup V_{\PR} \times
  V_{\PR}$ is a simulation relation for all $v_1=(S_1,{\cal P})$,
  $v_2=(S_2,{\cal P}')$: $v_1 \iesim v_2$ iff ${\cal P}={\cal P}'$
  and, for all $\tau_i\in\tau$, the three following conditions hold:
  \begin{inparaenum}[(i)]\item \label{item:pt1-def-iesim}
    $\rct_{S_1} (\tau_i) \geq \rct_{S_2}(\tau_i)$; and
  \item \label{item:pt2-def-iesim} $\rct_{S_2} (\tau_i) = 0$ implies
    $\rct_{S_1} (\tau_i) = 0$; and
  \item \label{item:pt3-def-iesim}
    $\nat_{S_1}(\tau_i)$ $\leq \nat_{S_2}(\tau_i)$.
  \end{inparaenum}
\end{definition}

\begin{theorem}[Taken from \cite{TCS}]
  $\iesim$ is a tba-simulation
\end{theorem}

\subsection{Computation of $\minac{\bad}$}
Given the $\iesim$ tba-simulation, let us now show how to compute
$\ALS_0=\minac{\Attr_0}=\minac{\bad}$ (i.e., the initial set in
\algorithmcfname~\ref{alg:bw-tba}), without computing first the whole
set $\bad$. Recall that, by definition $\bad$ contains only $\PR$
states (this is sufficient since, if a task misses a deadline at a
$\PS$ state, then all its predecessors---which are $\PR$ states---are
losing too).

We build the set $\minac{\bad}$ by considering each task
separately. Let us first consider a game containing only one task
$\tau_i$. Let $v$ be a state from $\bad$. Then,
$\laxity_v(\tau_i) = \nat_v(\tau_i) - T_i + D_i - \rct_v(\tau_i) < 0$.
Therefore, $v\in\bad$ iff $\rct_v(\tau_i)\in \{1,\ldots, C_i\}$ and
$\nat_v(\tau_i) \leq T_i - D_i + \rct_s(\tau_i) - 1$. We denote by
$\bad_{\iesim\tau_i}$ the set:
\begin{align*}
  \bad_{\iesim\tau_i} &=\{v \mid \exists 1\leq j\leq C_i: \nat_v(\tau_i) = T_i -
                        D_i + C_i - j\\
  &\phantom{=\{v \mid \exists 1\leq j\leq C_i:} \textrm{ and }\rct_v(\tau_i) = C_i - ( j +1 )\}.
\end{align*}
It is easy to check (see hereinafter) that, in this game with a single
task $\tau_i$, $\minac{\bad}=\bad_{\iesim\tau_i}$.

Let us consider now a scheduling game on a set
$\tau=\{\tau_1, \ldots, \tau_n\}$ of $n$ tasks. Let $v$ be a state
from $\bad$. Then, there exists at least one task $\tau_i \in \tau$
such that $\laxity_v(\tau_i) < 0$. Hence the compact representation of
all states that are bad because $\tau_i$ has missed its deadline is
$\bad_{\iesim\tau_i,\tau}$, where $v\in \bad_{\iesim\tau_i,\tau}$ iff
\begin{inparaenum}[(i)]
\item there is $B\in \bad_{\iesim\tau_i}$
  s.t. $\nat_B(\tau_i)=\nat_v(\tau_i)$ and
  $\rct_B(\tau_i)=\rct_v(\tau_i)$; and
\item
  for all $j\neq i$: $\nat_v(\tau_j)=T_j$ and
$\rct_v(\tau_j)\in\{0,1\}$.
\end{inparaenum}
Finally, we claim that the antichain covering all bad states is
$\bad_{\iesim} = \cup_{1\leq i\leq n}\bad_{\iesim\tau_i,\tau}$. 

\begin{lemma}
  \label{lem:bad} $\bad_{\iesim}$ is an antichain and
  $\bad = \uc{\bad_{\iesim}}$.
\end{lemma}
\begin{proof}
  Let us first prove that $\bad_{\iesim}$ is an antichain. Since
  $\forall \tau_i, \bad_{\iesim\tau_i}$ is an antichain,
  $\bad_{\iesim\tau_i,\tau}$ is either. Hence if $\bad_{\iesim}$ is
  not an antichain than there must exist
  $v_i \in \bad_{\iesim\tau_i,\tau}$ and
  $v_k \in \bad_{\iesim\tau_k,\tau}$ such that $v_i$ and $v_k$ are
  comparable w.r.t $\iesim$. By the definition above,
  $\nat_{v_i}(\tau_i) = T_i - D_i + C_i - j, \nat_{v_i}(\tau_k) = T_k$
  and
  $\nat_{v_k}(\tau_i) = T_i,\nat_{v_k}(\tau_k) = T_k - D_k + C_k - j'$
  where $j,j' \geq 1$. Since $C_i \leq D_i$ and $C_k \leq D_k$,
  $\nat_{v_i}(\tau_i) < \nat_{v_k}(\tau_i)$ and
  $\nat_{v_k}(\tau_k) < \nat_{v_i}(\tau_k)$. Then $v_i$ and $v_k$
  cannot be comparable (see Definition~\ref{def:idle}).

  For the second item, firstly, we will prove that
  $\bad \subseteq \uparrow \bad_{\iesim} $.  Given a state
  $s' \in \bad$, assume that $\tau_i \in \tau$ misses its deadline in
  $s$ We will prove that there exists a state $s\in \bad_{\iesim}$
  such that $s' \iesim s$.
\begin{inparaenum}[(i)]
\item For $\tau_i$, because
  $\laxity_{s'}(\tau_i) = \nat_{s'}(\tau_i) - T_i + D_i -
  \rct_{s'}(\tau_i) < 0$
  and $1 < \rct_{s'}(\tau_i) < C_i$. Then $\rct$ and $\nat$ of
  $\tau_i$ at $s$ is computed as follows:
  $\rct_s(\tau_i) = \rct_{s'}(\tau_i) $ and
  $\nat_s(\tau_i) = T_i - D_i + \rct_{s'}(\tau_i) - 1$
  (i.e. $(\nat_s(\tau_i), \rct_s(\tau_i)) \in
  \bad_{\iesim\tau_i}$).
  It is easy to see that
  $(\nat_{s'}(\tau_i), \rct_{s'}(\tau_i)) \iesim (\nat_s(\tau_i),
  \rct_s(\tau_i)) $
\item For $\tau'_{i}$, if $\rct_{s'}(\tau'_{i}) = 0$ than
  $\rct_{s}(\tau'_{i}) = 0$. And if $\rct_{s'}(\tau'_{i}) > 0$ than
  $\rct_{s}(\tau'_{i}) = 1$. In both the two cases,
  $\nat_{s}(\tau'_{i}) = T'_{i}$. Finally, we see that
  $(\nat_{s'}(\tau'_{i}), \rct_{s'}(\tau'_{i})) \iesim
  (\nat_s(\tau'_{i}), \rct_s(\tau'_{i})) $.
\end{inparaenum}

Now let us prove that $\bad \supseteq \uparrow \bad_{\iesim}$. Given
$v \in \uparrow \bad_{\iesim}$, then $\exists \tau_i \in \tau$
s.t. $\nat_s(\tau_i) \leq T_i - D_i + C_i - j$ and
$\rct_s(\tau_i) \geq C_i -( j -1)$ where $j = 1..C_i$. Then
$\laxity_s(\tau_i) \leq T_i - D_i + C_i - j - T_i + D_i - C_i -( j +
1)$.
Therefore, $\laxity_s(\tau_i) \leq -1$. In consequence $v \in \bad$.
\end{proof}

\subsection{Computation of $\UPreATC{\ATC}$}

Given this tba-simulation $\iesim$, let us now show how to compute
$\UPreATCname$. As for the definition of $\UPrename$, we split the
computation of uncontrollable predecessors according to the types of
the nodes, and let:
\begin{align*}
  \UPreATC{\ATC}&=\PreEATC{\ATC\cap V_\PS}\cup\PreAATC{\ATC\cap V_\PR},
\end{align*}
where the specifications for these two new operators are:
\begin{align*}
  \PreEATC{\ATC} &= \minac{\PreE{\uc{\ATC}}} &\textrm{for all
                                               antichains
                                               }\ATC\subseteq V_\PS,\\
  \PreAATC{\ATC} &= \minac{\PreA{\uc{\ATC}}} &\textrm{for all
                                               antichains
                                               }\ATC\subseteq V_\PR.
\end{align*}
\paragraph{Computation of $\PreEATCname$}
The computation of $\PreEATC{\ATC}$ for some antichain $\ATC$ of $\PS$
nodes turns out to be easy, as shown by the next lemma. Intuitively,
it says that, computing $\PreE{U}$ for some upward-closed set
$U\subseteq V_\PR$ which is given to us by its compact representation
(minimal antichain) $\ATC=\minac{U}$, boils down to computing
$\PreE{\ATC}$ (i.e., computing the predecessors of the minimal
elements). After minimisation, $\minac{\PreE{\ATC}}$ gives us the
compact representation of $\PreE{U}$. Observe that this Lemma holds
for any tba-simulation $\wbs$, and not only for $\iesim$:
\begin{lemma}
\label{lem:prea}
For all antichains $\ATC\subseteq V_\PS$, for all tba-simulations
$\wbs$: $\PreEATC{\ATC}=\minac{\PreE{\ATC}}$
\end{lemma}
\begin{proof}
   To show this lemma, one must show that
  $\minac{\PreE{\uc{\ATC}}}=\minac{\PreE{\ATC}}$, since by definition of
  $\PreEATCname$,
  $\PreEATC{\ATC}=\minac{\PreE{\uc{\ATC}}}$. We thus prove that:
  \begin{align}
    \PreE{\uc{\ATC}} &= \uc{\PreE{\ATC}}\label{eq:1}
  \end{align}

  We show the inclusion in both directions. 
  
  First,
  $\PreE{\uc{ \ATC}} \subseteq \uc{\PreE{\ATC}}$.  For all
  $v \in \PreE{\uc{\ATC}}$, there exists $v'$ s.t.
  $v' \in \Succ{v} \cap \uc{\ATC}$, by definition of $\PreEname$. Then
  $v' \in \uc{\Succ{v} \cap \ATC}$. Hence
  $\Succ{v} \cap \ATC \neq \emptyset$. Therefore $v \in \PreE{\ATC}$
  (see equation \eqref{eq:pree1}), in consequence $v \in \uc{\PreE{\ATC}}$.
  Next, let us show that
  $\uc{\PreE{\ATC}} \subseteq \PreE{\uc{ \ATC}}$. For each
  $v \in \uparrow \PreE{\ATC} $, there exists $v'$ s.t.
  $v' \in \PreE{\ATC}$ and $v \wbs v'$. There hence exists
  $\overline{v'}$, a successor of $v'$ and $v' \in \ATC$ (since
  $\Succ{v'} \cap \ATC \neq \emptyset$). By Definition~\ref{def:tba},
  there exists $\overline{v} \in \Succ{v}$ such that
  $\overline{v} \wbs \overline{v'}$. Since $\overline{v'} \in \ATC$
  then $\overline{v} \in \uparrow \ATC$. Hence,
  $v \in \PreE{\uc{ \ATC}}$.
\end{proof}

It remains to observe that, given any set $V'\subseteq V_{\PS}$,
computing $\minac{\PreE{V'}}$ can be done directly (i.e., without the
need of building and exploring a large portion of the game graph), by
`inverting' the definition of the successor relation, and making sure
to keep minimal elements only. More precisely, let
$v=(S,\PS)\in V_\PS$ be a \PS node. Then,
$v'=(S',\PR)\in \PreE{\{v\}}$ iff the following conditions hold for
all $\tau_i \in \tau$:
\begin{inparaenum}[(i)]
\item if $\rct_S(\tau_i) < C_i$, then
  $\rct_S(\tau_i) = \rct_{S'}(\tau_i)$ and
  $\nat_S(\tau_i) = \nat_{S'}(\tau_i)$;
and \item if $\rct_S(\tau_i) = C_i$ than either $\rct_{S'}(\tau_i) = 0$
  and $\nat_{S'}(\tau_i) = 0$ or $\rct_{S'}(\tau_i) = C_i$ and
  $\nat_{S'}(\tau_i) = \nat_{S}(\tau_i)$.
\end{inparaenum}
In the end, $\minac{\PreE{V'}}$ = $\minac{\bigcup_{v \in V'} \PreE{\{v\}}}$.

\paragraph{Computation of $\PreAATCname$} Let us now show how to
compute, in an efficient fashion, the operator
$\PreAATCname$. Remember that we are given an antichain
$\ATC\subseteq V_\PR$ which is a compact representation of some
upward-closed set $\uc{\ATC}$, and that we want to compute
$\minac{\PreA{\uc{\ATC}}}$, without needing to compute the full
$\uc{\ATC}$ set. Unfortunately, and contrary to $\PreEATCname$, it is
not sufficient to consider the predecessors of the elements in $\ATC$
to deduce $\PreAATC{\ATC}$, as shown by the example in
\figurename~\ref{fig:preA-difficult}. In this example, the antichain
$\ATC$ is $\{\PR_1,\PR_2\}$. The gray cones depict the upward-closures
of those states. Now, if we consider the predecessors of
$\PR_1\in\ATC$, we obtain $\PS_1$, however this state has a successor
$\PR_3$ which is not in $\uc{\ATC}$, so
$\PS_1\not\in\PreA{\uc{\ATC}}$, hence $\PS_1\not\in\PreAATC{\ATC}$.
Yet, on this example, there is a state $\PS_1'\iesim\PS_1$ which has
all its successors in $\uc{\ATC}$, but which is neither in
$\PreE{\ATC}$ nor in $\PreA{\ATC}$.

\begin{figure}
  \centering

  \begin{tikzpicture}[scale=.75]
    \draw[gray!60, fill=gray!60, opacity=.5] (5,3) -- (7,7) -- (3.5,7) -- cycle ;
    \draw[gray!60, fill=gray!60, opacity=.5] (3,5) -- (5,7) -- (1.5,7) -- cycle ;
    \node[draw,rectangle,fill=white] (T3) at (4,2) {$\PR_3$} ;
    \node[draw,rectangle,fill=white] (T1) at (5,3) {$\PR_1$} ;
    \node[draw,rectangle,fill=white] (T2) at (3,5) {$\PR_2$} ;
    \node[draw,rectangle,fill=white] (T1p) at (5,5) {$\PR_1'$} ;
    \node[draw,rectangle,fill=white] (T2p) at (3,6) {$\PR_2'$} ;
    \node[draw,circle] (S1) at (2,2) {$\PS_1$} ;
    \node[draw,circle] (S1p) at (1,4) {$\PS_1'$} ;
    \path[->] (S1) edge (T3)
                   edge[bend left] (T1)
              (S1p) edge[bend right=15] (T1p)
                    edge[bend left] (T2p) ;
    \path[dotted] (S1) edge[bend left] node[rotate=-70] {$\iesim$}
    (S1p) ;
    \path[dotted] (T2p) edge[bend left] node[rotate=-70] {$\iesim$} (T3) ;
  \end{tikzpicture}
  \caption{$\PreAATC{\{\PR_1,\PR_2\}}$ cannot be computed by
    considering only the predecessors of $\PR_1$ and $\PR_2$.}
  \label{fig:preA-difficult}
\end{figure}
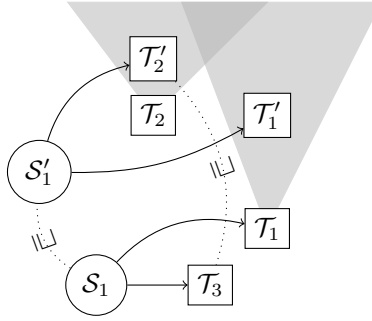

Unfortunately, we haven't managed to find a \emph{direct} way to
compute $\PreAATC{\ATC}$ from $\ATC$; i.e., without enumerating some
elements of $\uc{\ATC}$ that are not in $\ATC$. We propose
a method that, while not working only on elements of $\ATC$,
avoids the full enumeration of $\uc{\ATC}$, and performs well in
practice. It is described in \algorithmcfname~\ref{alg:upre}. Instead
of computing $\PreAATC{\ATC}$, it returns the subset
$\PreAATC{\ATC}\setminus \uc{\ATC}$. Observe that this subset is
sufficient for \algorithmcfname~\ref{alg:bw-tba}, because, at each
step, one computes, from $\ALS_i$, the set $\ALS_{i+1}$ as:
\begin{align}
 \ALS_{i+1} = \lfloor &\PreEATC{\ALS_i} \cup\PreAATC{\ALS_i} \nonumber \\
 & \cup\ALS_i \rfloor. \label{eq:2}
\end{align}
 
 \begin{algorithm}
  \DontPrintSemicolon
  \SetKwFunction{UP}{Compute-$\PreAname$}
  \UP\Begin{
    
    $\mathcal{B}$ $\leftarrow$ $\emptyset$ \;
    $\mathrm{ToSearch}$ $\leftarrow$ $\emptyset$ \;
    
    \ForEach{$v \in \PreE{\ATC}\setminus\uc{\ATC}$} {\label{alg:init}
      \lIf{$\Succ{v} \subseteq \uc{\ATC} $}{
        $\mathcal{B}$ $\leftarrow$ $\minac{\mathcal{B} \cup \{v\} }$ 
      }
      \lElse{
        $\mathrm{ToSearch}$ $\leftarrow$ $\mathrm{ToSearch}\cup\uc{v}\setminus\uc{\ATC}$ 
      }
    }

    \While{$\mathrm{ToSearch}\neq\emptyset$}{\label{alg:compute}
      Pick and remove $v$ from $\mathrm{ToSearch}$ \;
      \lIf{$\Succ{v} \subseteq \uc{\ATC}$}{
        \label{alg:to-ac}
        $\mathcal{B}$ $\leftarrow$ $\minac{\mathcal{B} \cup \{v\} }$
      }
      \lElse{\label{alg:optim-dc}
        $\mathrm{ToSearch}$ $\leftarrow$ $\mathrm{ToSearch}\setminus
        \{v'\mid v\iesim v'\}$ 
      }
    }

    \Return{$\mathcal{B}$} \;
  }
  \caption{Computing
    $\PreAATC{\ATC}\setminus \uc{\ATC}$.}

  \label{alg:upre}
 
\end{algorithm}

Hence, it is easy to check that replacing $\PreAATC{\ALS_i}$ in
\eqref{eq:2} by:
\[
  \PreAATC{\ALS_i}\setminus \uc{\ALS_i}
\]
does not change the value of the expression. Yet, computing
$\PreAATC{\ATC}\setminus \uc{\ATC}$ instead of $\PreAATC{\ATC}$ allows
us to avoid considering elements which are already computed (in
$\ATC$), which makes \algorithmcfname~\ref{alg:upre} efficient in
practice.  This algorithm works as follows. First, we build a set
$\mathrm{ToSearch}$ of candidates nodes for the set $\PreAATC{\ATC}$,
then we explore it. $\mathrm{ToSearch}$ is initialised by the loop
starting at line~\ref{alg:init}. In this loop, we consider all the
predecessors $v$ of $\ATC$ (we will explain hereinafter how these
predecessors can be computed efficiently). We keep only the
predecessors $v$ that are not covered by $\ATC$, and check whether all
successors of $v$ are in $\uc{\ATC}$. If yes, $v\in\PreA{\uc{\ATC}}$,
and we add $v$ to the antichain $\mathcal{B}$ (that will eventually be
returned). Remark that the test $\Succ{v} \subseteq \uc{\ATC} $ can be
performed without computing explicitly $\uc{\ATC} $, by checking that,
for all $v'\in\Succ{v}$, there is $v''\in \ATC$ s.t. $v'\iesim v''$.
If there are some successors of $v$ that are not in $\uc{\ATC}$, we
must consider all nodes larger than $v$ as candidates for
$\PreAATC{\ATC}$, but not those that are already in $\uc{\ATC}$,
hence, we add to $\mathrm{ToSearch}$ the set
$\uc{v}\setminus\uc{\ATC}$. The fact that we add
$\uc{v}\setminus\uc{\ATC}$ instead of $\uc{v}$ is important to keep
$\mathrm{ToSearch}$ small, which is a key optimisation of the
algorithm. Then, once $\mathrm{ToSearch}$ has been built, we examine
all the elements $v$ it contains, and check whether they belong to
$\PreA{\uc{\ATC}}$, in the loop starting in line~\ref{alg:compute}. If
$v\in\PreA{\uc{\ATC}}$ (line~\ref{alg:to-ac}), $v$ is added to
$\mathcal{B}$. Otherwise, another optimisation of the algorithm
occurs: $v$ is discarded from $\mathrm{ToSearch}$, as well as
\emph{all the nodes that are smaller than} $v$ (line 10). This is
correct by properties of the tba-simulation (see
\cite{DBLP:conf/rp/GeeraertsGS14}), and allows the algorithm to
eliminate candidates from $\mathrm{ToSearch}$ much quicker.

To conclude the description of \algorithmcfname~\ref{alg:upre}, let us
explain how to compute $\PreE{\ATC}$ efficiently, when
$\ATC\subseteq V_\PR$, as needed in line~\ref{alg:init}. Let $(S,\PR)$
be a $\PR$-node. Then, $(S',\PS)\in\PreE{S,\PR}$ iff there is a set
$\tau'\subseteq\tau$ of scheduled tasks s.t. $|\tau'|\leq m$ (at most
$m$ tasks have been scheduled) and the following conditions hold for
all $\tau_i \in \tau$:
\begin{inparaenum}[(i)] 
\item $\nat_{S'}(\tau_i) = \min(\nat_{S}(\tau_i) + 1,T_i)$
\item $\rct_S(\tau_i) = C_i$ implies $\rct_{S'}(\tau_i) = \rct_S(\tau_i)  $
\item
  $\rct_S(\tau_i) < C_i$ implies $\rct_{S'}(\tau_i) = \rct_S(\tau_i)$
  if $\tau_i\not\in \tau'$; and $\rct_{S'}(\tau_i) = \rct_S(\tau_i) +
  1$ otherwise.
\end{inparaenum}

\subsection{Further optimisation of the algorithm}\label{sec:furth-optim-algor}
We close our discussion of \Bw by describing the last optimisation that
further enhances the performance of the algorithm. In
\algorithmcfname~\ref{alg:bw-tba}, the $\UPreATCname$ operator is
applied, at each step of the main loop, to the whole $\ALS_i$ set in
order to compute $\ALS_{i+1}$. While this makes the discussion of the
algorithm easier, it incurs unnecessary computations: instead of
computing the uncontrollable predecessors of \emph{all} states in
$\ALS_i$, one should instead compute the uncontrollable predecessors
of the node that have been inserted in $\ALS_i$ at the previous
iteration of the loop \emph{only}. As an example, on the game in
\figurename~\ref{fig:gameRS},  
$\ALS_0=\minac{\{\PR_4,\ldots,\PR_{10}\}}$;
$\ALS_1=\ALS_0\cup \{\PS_2\}$ since $\PS_2$ is incomparable to
$\PR_4,\ldots,\PR_{10}$. Then, to discover that $\PR_2$ is a losing
node too, we should only apply $\UPreATCname$ to $\PS_2$, and not to
$\PR_4,\ldots,\PR_{10}$ again.

To formalise this idea, we introduce the notion of
\emph{frontier}. Intuitively, at each step $i$ of the algorithm,
$\Frontier_i$ is an antichain that represents the {\it newly
  discovered} losing states, i.e. the states that have been declared
losing at step $i$ but were not known to be losing at step $i-1$. To
define formally the sequence $\left(\Frontier_i\right)_{i\geq 0}$, we
introduce the operator $\PreAS{\mathcal{A}}{\mathcal{C}}$ defined as
$\PreAS{\mathcal{A}}{\mathcal{C}}=\break\uc{\PreE{\mathcal{C}}}\cap
\PreAATC{\mathcal{A}}$.
That is, nodes in $\PreAS{\mathcal{A}}{\mathcal{C}}$ have all their
successors in $\mathcal{A}$ and are larger than some predecessor of
$\mathcal{C}$. Hence, $\PreAS{\mathcal{A}}{\mathcal{C}}$ can be
computed by adapting \algorithmcfname~\ref{alg:upre}, substituting
$\PreE{\mathcal{C}}$ for $\PreE{\ATC}$ in
line~\ref{alg:init}. Clearly, if we manage to keep the size of
$\mathcal{C}$ as small as possible, we will further decrease the size
of $\mathrm{ToSearch}$, and \algorithmcfname~\ref{alg:upre} will be
even more efficient.

Now let us define the sequence of $\Frontier_i$ sets. We let
$\Frontier_0 = \minac{\bad}=\bad_{\iesim}$. For all $i > 0$, we let:
\[
  \Frontier_{i+1}= \minac{\UPreS{\ALS_i}{\Frontier_i}},
\]
and:
\begin{eqnarray*}
  &\UPreS{\ALS_i}{\Frontier_i}\\
  &=\\
  &\PreE{\Frontier_i \cap V_\PS} \cup
  \PreAS{\ALS_i}{\Frontier_i \cap V_\PR}.
\end{eqnarray*}

Observe that in this new definition of the uncontrollable predecessor
operator, only predecessors of nodes in the frontier are considered
for computation, which is the key of the optimisation. The next lemma
and proposition justify the correctness of this new construction. In
particular, Proposition~\ref{prop:frontier} shows that the sequence
$\left(\ALS_i\right)_{i\geq 0}$ (which is computed by \Bw) can be
computed applying the operator we have just described, on elements
from the frontier only.

\begin{figure}
\center
 \includegraphics[width=0.8\textwidth]{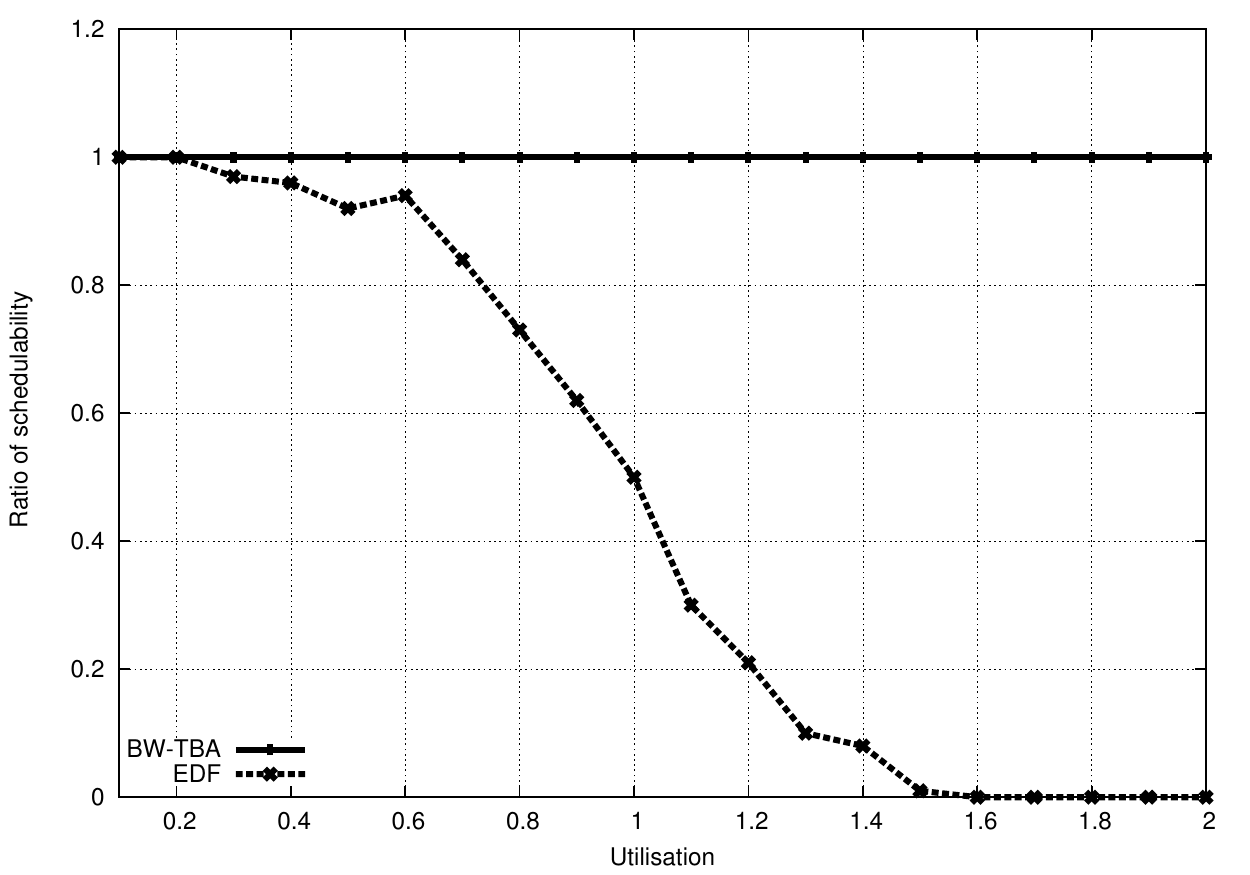}
 \caption{Ratio of schedulable systems by $\EDF$ and by $\Bw$.}
 \label{fig:edf_bw}
\end{figure}

To establish Proposition~\ref{prop:frontier}, we first need an
ancillary lemma, showing that all the elements that occur in the
frontier at some step $i$ are also in $\ALS_i$:

\begin{lemma}
\label{lem:frontier}
For all $i \in \mathbb{N}, \Frontier_i \subseteq \ALS_i$.
\end{lemma}

\begin{proof}
  We prove the lemma by induction on $i$.

  {\bf Base case:} When $i = 0$, $\ALS_0 = \bad_\iesim$ and
  $\Frontier_0$ = $\bad_\iesim$. Hence the lemma is trivially correct.

  {\bf Inductive case: } Given $\Frontier_i \subseteq \ALS_i$, we will
  prove that: 
  \[
    \Frontier_{i+1} \subseteq \ALS_{i+1},
  \] i.e.  $\forall v \in \Frontier_{i+1}: v \in \ALS_{i+1}$.

  By the definition of the sequence $\Frontier$:
  \begin{align*}
    \Frontier_{i+1} &=\lfloor \UPreS{\ALS_i}{\Frontier_i} \rfloor \\
    &=\lfloor \PreE{\Frontier_i \cap V_\PS} \cup \PreAS{\ALS_i}{\Frontier_i
    \cap V_\PR} \rfloor.
  \end{align*}
 
  By the definition of the sequence $\ALS$, 
  $\ALS_{i+1} = \lfloor \ALS_i \cup \PreE{\ALS_i \cap V_\PS} \cup
  \PreAATC{\ALS_i \cap V_\PR} \rfloor $.
  Hence, proving that the two following items are correct is sufficient.
  \begin{enumerate}[(i)]
  \item  $\PreE{\Frontier_i \cap V_\PS} \subseteq \PreE{\ALS_i \cap V_\PS} $ 
  \item  $\PreAS{\ALS_i}{\Frontier_i \cap V_\PR} \subseteq \PreAATC{\ALS_i \cap V_\PR} $.
  \end{enumerate}
  We then prove these two items as follows:
  \begin{itemize}[(i)]
  \item Given $v \in \PreE{\Frontier_i \cap V_\PS}$, then 
    $\exists v' \in \Succ{v}$ where $v' \in \Frontier_i \cap V_\PS$ (see equation
    \eqref{eq:pree2}). Since $\Frontier_i \subseteq \ALS_i$ then
    $v' \in \ALS_i$ as well. Hence, $v \in \PreE{\ALS_i \cap V_\PS}$.
  \item Given $v \in \PreAS{\ALS_i}{\Frontier_i \cap V_\PR} $ then
    $v \in \PreE{\Frontier_i \cap V_\PR} \cap \PreAATC{\ALS_i \cap V_\PR}$.
    Therefore,
    $v \in \PreAATC{\ALS_i \cap V_\PR} $ as well.
  \end{itemize}
  Hence the lemma is successfully proved. 
\end{proof}

We can now prove Proposition~\ref{prop:frontier}:

\begin{proposition}
\label{prop:frontier}
For all $i \in \mathbb{N}$:
\[
\ALS_{i+1} = \minac{\ALS_{i} \cup \Frontier_{i+1}}.
\]
\end{proposition}
\begin{proof} [Proof of Proposition~\ref{prop:frontier}]
  The proof is proved by induction on $i$.

  {\bf Base case: } When $i=0$, $\ALS_0 = \Frontier_0 = \bad_\iesim$. By definition of the
  sequence $\ALS$:
  \begin{align*}
    \ALS_1 &= \lfloor \ALS_0 \cup \UPreATC{\ALS_0} \rfloor\\
    \UPreATC{\ALS_0} &=  \PreE{\ALS_0 \cap V_\PS} \cup \PreAATC{\ALS_0 \cap V_\PR}.
  \end{align*}
  On the other hand:
  \begin{align*}
    \Frontier_1 &= \lfloor \UPreS{\Frontier_0}{\Frontier_0} \rfloor\\
    &= \lfloor \PreE{\Frontier_0 \cap V_\PS} \cup \PreAS{\ALS_0} {\Frontier_0
    \cap V_\PR} \rfloor.
  \end{align*}  
  Since $\ALS_0 = \Frontier_0$, then:
  \[
    \PreAATC{\ALS_0} = \PreAS{\ALS_0} {\Frontier_0}.
  \]
  Hence,
  $\UPreATC{\ALS_0} = \Frontier_1$.  Finally,
  $\ALS_{1} = \lfloor \ALS_{0} \cap \Frontier_{1} \rfloor$.

  {\bf Inductive case: } For all $i>0$, given
  $\ALS_{i+1} = \lfloor \ALS_{i} \cup \Frontier_{i+1} \rfloor$, we need
  to prove 
  $\ALS_{i+2} = \lfloor \ALS_{i+1} \cup \Frontier_{i+2} \rfloor$. In
  other words we need to prove that:
  \begin{eqnarray*}
    &\lfloor \ALS_{i+1} \cup \PreE{\ALS_{i+1} \cap V_\PS} \cup
      \PreAATC{\ALS_{i+1} \cap V_\PR} \rfloor\\
    &=\\
    &\left \lfloor
      \begin{array}{c}
        \ALS_{i+1} \cup
      \PreE{\Frontier_{i+1} \cap V_\PS}\\ \cup\\
      \PreAS{\ALS_{i+1}}{\Frontier_{i+1} \cap V_\PR}
      \end{array}
      \right\rfloor.    
  \end{eqnarray*}

  We divide into two cases as follows.
  \begin{enumerate}
  \item For all states controlled by player \PS,  the proposition
    becomes:
    \begin{eqnarray*}
      &\lfloor \ALS_{i+1} \cup \PreE{\ALS_{i+1}} \rfloor\\ &=\\ 
      &\lfloor \ALS_{i+1} \cup \PreE{\Frontier_{i+1}}
        \rfloor.
    \end{eqnarray*}
    We need to prove that $\ALS_{i+1} \cup \PreE{\ALS_{i+1}}$
    (abbreviated by $\mathcal{L}$) =
    $\ALS_{i+1} \cup \PreE{\Frontier_{i+1}}$ (abbreviated by
    $\mathcal{R}$).
    \begin{itemize}
    \item Assume that $\exists v \in \mathcal{L}$ such that
      $v \not\in \mathcal{R}$ ($v \in V_\PR$).  Hence:
      \begin{inparaenum}[(i)]
      \item \label{item:1} $v \not\in \ALS_{i+1}$; and
      \item \label{item:2} $v \in \PreE{\ALS_{i+1} }$.
      \end{inparaenum}
      From~\ref{item:1}, there does not exist $v' \in \ALS_{i}$ such
      that $v \in \PreE{\{v'\}}$. From~\ref{item:2}, there exists
      $v' \in \ALS_{i+1}$ such that $v \in \PreE{\{v'\}}$. Hence
      $v' \in \ALS_{i+1} \setminus \ALS_i$, i.e.
      $v' \in \Frontier_{i+1}$ (thanks to the hypothesis). In
      consequence, $v \in \PreE{\Frontier_{i+1}}$. In
      contradiction. Therefore $\mathcal{L} \subseteq \mathcal{R}$.

    \item Given $v \in \mathcal{R}$, if $v \in \ALS_{i+1}$ then
      $v \in \mathcal{L}$. Otherwise,
      $v \in \PreE{\Frontier_{i+1}}$. Hence there exists
      $v' \in \Frontier_{i+1}$ such that $v \in \PreE{\{v'\}}$. Since
      $\Frontier_{i+1} \subseteq \ALS_{i+1}$ (see
      Lemma~\ref{lem:frontier}), $v' \in \ALS_{i+1}$. Thus by equation
      $\eqref{eq:pree2}$, $v \in \PreE{\ALS_{i+1}}$, thus
      $v \in \mathcal{L}$. Hence,
      $\mathcal{L} \supseteq \mathcal{R}$.
    \end{itemize}

  \item For all states controlled by player \PR, the proposition
    becomes:
    \begin{eqnarray*}
      &\lfloor \ALS_{i+1} \cup \PreAATC{\ALS_{i+1}} \rfloor\\ &=\\
      &\lfloor \ALS_{i+1} \cup \PreAS{\ALS_{i+1}}{\Frontier_{i+1}} \rfloor.
    \end{eqnarray*}
    We need to prove
    $\ALS_{i+1} \cup \PreAATC{\ALS_{i+1}} = \ALS_{i+1} \cup
    \PreAS{\ALS_{i+1}}{\Frontier_{i+1}}$. For the sake of clarity let
    us denote:
    \begin{align*}
      \mathcal{L} &= \ALS_{i+1} \cup \PreAATC{\ALS_{i+1}}\\
      \mathcal{R} &= \ALS_{i+1} \cup
                    \PreAS{\ALS_{i+1}}{\Frontier_{i+1}}.
    \end{align*}
    Hence, we need to show that $\mathcal{L}=\mathcal{R}$.

    \begin{itemize}
    \item Assume that there exists a state $v \in \mathcal{L}$ such
      that $v \not\in \mathcal{R}$. Therefore $v \not\in \ALS_{i+1}$
      (there exists $v' \not\in \ALS_i$ where $v \in\break \uc{\PreE{v'}}$
      (a)) and $v \in \PreAATC{\ALS_{i+1}}$, therefore
      $\Succ{v} \subseteq\break \uc{\ALS_{i+1}}$ (b), i.e.
      $v' \in \ALS_{i+1}$. Thanks to the hypothesis,
      $v' \in \ALS_{i+1} \setminus \ALS_i$, i.e.
      $v' \in \Frontier_{i+1}$ (c). By combining (a), (b) and (c),
      $v \in \PreAS{\ALS_{i+1}}{\Frontier_{i+1}}$. In consequence
      $v \in \mathcal{R}$. In contradiction. Hence,
      $\mathcal{L} \subseteq \mathcal{R}$.

    \item Given $v \in \mathcal{R}$, if $v \in \ALS_{i+1}$ than
      $v \in \mathcal{L}$. Otherwise,
      $v \in \PreAS{\ALS_{i+1}}{\Frontier_{i+1}}$, therefore
      $v \in \PreAATC{\ALS_{i+1}}$ by definition of the operator $\PreAS{\mathcal{A}}{\mathcal{B}}$. Hence,
      $\mathcal{L} \supseteq \mathcal{R}$.
    \end{itemize}
  \end{enumerate}
  We conclude that
  $\ALS_{i+1} = \lfloor \ALS_{i} \cup \Frontier_{i+1} \rfloor$.
\end{proof}

Hence, the formula in line 5 of 
\algorithmcfname~\ref{alg:bw-tba} can be replaced by
$\ALS_{i+1} = \minac{\ALS_{i} \cup \Frontier_{i+1}}$ (see Proposition~\ref{prop:frontier},), where
$\Frontier_{i+1}$ is computed from $\Frontier_i$: 
$\Frontier_{i+1} = \lfloor\PreAS{\ALS_i}{\Frontier_i\cap V_\PR} \cup \ALS_{i}\cup$ \\$\PreE{\Frontier_i\cap V_\PS}\rfloor$.
Since all the computations are now performed on the frontier instead of
the whole $\ALS_i$ set, the performance of
\algorithmcfname~\ref{alg:bw-tba} is improved. This is the algorithm
we have implemented and that we report upon in the next section.

\section{Experimental Results\label{sec:experimental-results}}
\label{sec:experimental-results}

Let us now report on a series of experiments that demonstrate the
practical interest of our approach. Remember that our algorithm
performs an {\it exact test}, i.e. given any real-time system of
sporadic tasks, the algorithm always computes an online scheduler if
one exists. Otherwise, the algorithm proves the absence of scheduler
for the system.

Our implementations are made in C++ using the STL. We performed our
experiments on a Mac Pro (mid 2010) server with OS X Yosemite,
processor of 3.33 GHz, 6 core Intel Xeon and memory of 32 GB 1333 MHz
DDR3 ECC. Our programs were compiled with Apple Inc.'s distribution of
g++ version 4.2.1.

We compare the performance of our improved algorithm $\Bw$
(Algorithm~\ref{alg:bw-tba}, implemented with the frontier
optimisation described in Section~\ref{sec:furth-optim-algor}) to
three other approaches. The first one consists in scheduling the
system with the classical $\EDF$ (Earliest Deadline First) scheduler.
The second (called $\ES$) consists in first building the portion of
the game graph that contains all states that are reachable from $I$
(the initial state), then applying Algorithm~\ref{alg:es} to compute
\losing, and finally testing whether $I\in\losing$ or not. The third
approach, called $\Fw$ has been introduced in \cite{TCS}. Contrary to
$\Bw$ and $\ES$ that are \emph{backward approaches} (since we unfold
the successor relation in a backward fashion, starting from the \bad\
states), $\Fw$ is a \emph{forward algorithm}. It builds a portion of
the game graph and looks for a winning strategy using on-the-fly
computation. The idea has originally been introduced as a general
algorithm called OTFUR for safety games \cite{CDFLL-concur05}, and
$\Fw$ is essentially an improved version of OTFUR using antichain
techniques (and the $\iesim$ partial order), just as $\Bw$ is an
improved version of $\ES$.  In \cite{TCS}, we report on experiments
showing $\Fw$ outperforms $\ES$ in practice.


$\EDF$ is known to be {\it optimal} for arbitrary collections of
independent jobs scheduled upon \emph{uni}processor platforms. This
optimality result, however, is no longer true on \emph{multi}processor
platforms (that is, there are systems of sporadic tasks for which an
online scheduler exists, but that will not be scheduled by
$\EDF$). The goal of our first experiment is to compare the number of
systems that $\EDF$ and $\Bw$ can schedule. To this aim, we have
generated 2,000 sets of tasks, grouped by values of utilisation
$U$. For all instances, we consider $m=2$ CPUs and $n=4$ tasks. For
each $U \in [0.1,2]$ with a step of 0.1, 100 instances are randomly
generated. The $T$ parameters of all tasks range in the interval
$[5,15]$. The generation of this benchmark is based on the
\textsc{Uunifast} algorithm \cite{DBLP:journals/rts/BiniB05}.

\figurename~\ref{fig:edf_bw} depicts the ratio of systems scheduled by
$\EDF$ and $\Bw$ on the benchmark described above. All of the 2,000
randomly generated instances are feasible, hence $\Bw$ always computes
a scheduler, while the ratio of systems that are schedulable by $\EDF$
decreases sharply with the increase in the utilisation factor. This is
a first strong point in favour of our approach.

\begin{figure}
\centering
 \includegraphics[width=0.8\textwidth]{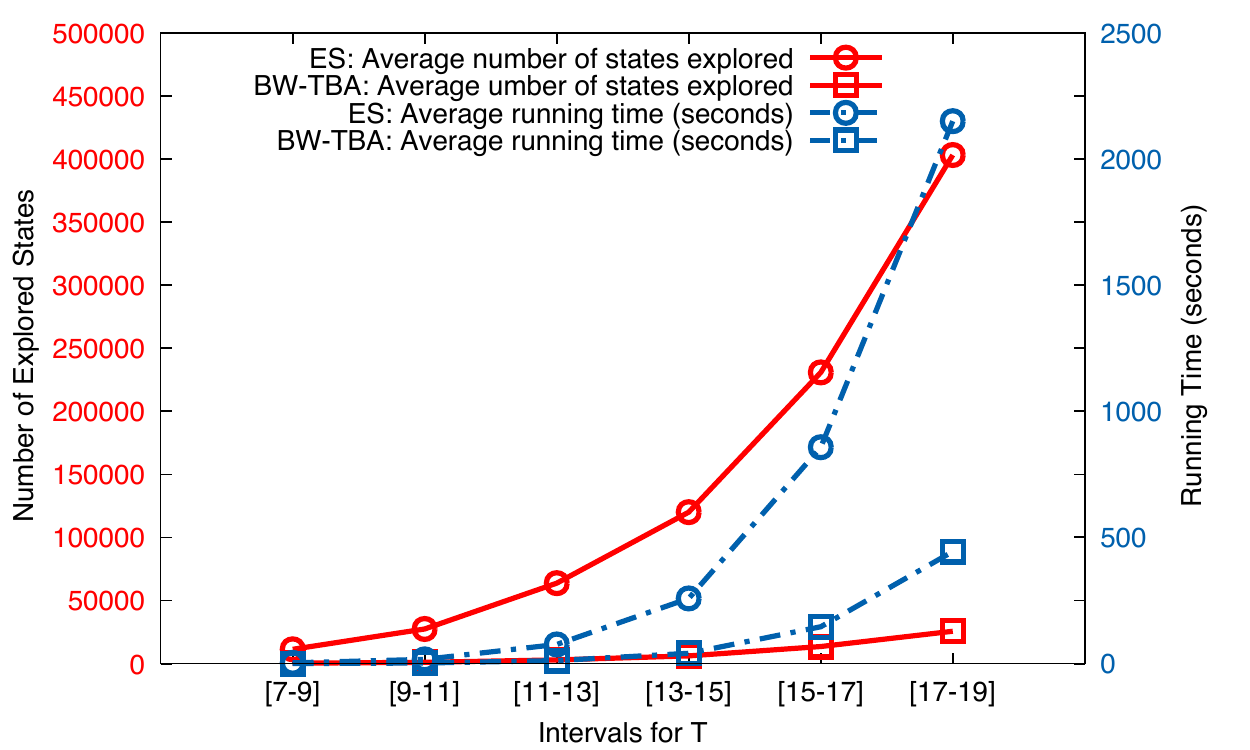}
 \caption{Comparison of the average space explored and the running
   time between $\ES$ vs. $\Bw$ for each $\tau^{[i, i+2]}$. }
 \label{fig:es}
\end{figure}
 
Like $\Bw$, the $\ES$ algorithm performs an exact feasibility test,
but we naively builds the whole graph of the game which is of
exponential size. In order to compare the performance of $\ES$ and
$\Bw$, we reuse the benchmark presented in \cite{TCS} (also based on
\textsc{Uunifast}). It consists of 2,100 sets of tasks grouped by
values of $T$, the interarrival time.  For all instances, we consider
$m = 2$ CPUs and $n=3$ tasks. For each $i$ in
$\{7,9,11,13,15,17,19\}$, a set called $\tau^{[i,i+2]}$ of 300 tasks
is generated, with a $T$ parameter in the $[i,i+2]$ range.  For each
game instance in the benchmark, we run both $\Bw$ and $\ES$ and
collect two metrics: the number of explored states and the running
time.

\figurename~\ref{fig:es} displays, for each set $\tau^{[i,i+2]}$, the
average number of states explored and the average running time, of the
two algorithms. Our improved algorithm $\Bw$ outperforms $\ES$ both in
terms of space and running time, by approximately one order of
magnitude. These results definitely demonstrate the practical interest
of antichain techniques for scheduling games.

Now, let us compare the performance of the forward and backward
algorithms both enhanced with antichains. We performed another
experiment on a benchmark where we let the number of tasks vary. We
randomly generated 45 instances of systems with $n$ tasks for
$n\in\{3,4,5, 6\}$ (on a 2 CPUs platform). In all instances, we have
$T\in[5,7]$, $U \in \{1,1.5,2\}$ and
$D \in [C,T]$. Like the two previous benchmarks, this one is 
based on \textsc{Uunifast} too.

\begin{table}[h]
\small
\begin{center}
  \begin{tabular}{crrrr}
    \toprule
    {\bf Ratio}     
    & \multicolumn{1}{c}{$\boldsymbol{n=3}$}  
    & \multicolumn{1}{c}{$\boldsymbol{n=4}$}  
    & \multicolumn{1}{c}{$\boldsymbol{n=5}$} 
    & \multicolumn{1}{c}{$\boldsymbol{n=6}$}    \\ 
    \midrule
      $\boldsymbol{(0,10\%)}$ & 0\% & 4.44\% & 46.67\% & 64.44\%\\ 
      $\boldsymbol{[10\%,50\%)}$ & 88.88\% &95.56\% &53.33\% & 35.56\%\\
      $\boldsymbol{ (50\%,100\%)}$  & 11.12\% & 0 \% & 0 \% & 0\%\\ \bottomrule
  \end{tabular}

\end{center}
\caption{Percentage of systems with $n$ tasks for which the ratio on
  the number of explored nodes ($\Bw$/$\Fw$) falls in the $(0,10\%)$,
  $[10\%,50\%)$ or $(50\%,100\%)$ interval.}
 \label{tab:space}
\end{table}

Our experiments show that, on all instances, $\Bw$ explores less
states than $\Fw$. This phenomenon happens because $\Bw$ computes only
the losing states while $\Fw$ explores both losing and winning states
during search.  \tablename~\ref{tab:space}, presents the ratio of
space explored by $\Bw$ and $\Fw$, for example, on systems with 6
tasks, $\Bw$ consumes less than 10\% of the memory needed by $\Fw$ on
64.64\% of instances.  Note that the memory performance of $\Bw$ (vs
$\Fw$) improve when the number of tasks increase.
  
\begin{table}[h]
\small
\begin{center}
  \begin{tabular}{crrrr}
    \toprule
     {\bf Num. tasks}   
    & \multicolumn{1}{c}{$\boldsymbol{n=3}$}  
    & \multicolumn{1}{c}{$\boldsymbol{n=4}$}  
    & \multicolumn{1}{c}{$\boldsymbol{n=5}$} 
    & \multicolumn{1}{c}{$\boldsymbol{n=6}$}    \\ 
    \midrule
       {\bf Ratio} &  84.44\% & 64.44\% & 44.44\%  & 37.77\% \\ 
      \bottomrule
  \end{tabular}
\end{center}
\caption{The ratio on the number of systems for which \Bw runs faster
  than \Fw.}
 \label{tab:time}
\end{table}

\tablename~\ref{tab:time} presents the ratio of instances on which
$\Bw$ runs faster than $\Fw$. Notice that the time taken for each
instance (of a fixed number of tasks) is variable.  For example, on
6-task systems, some are done within one minute; while others take
hours to complete. The time taken depends on several hardly
predictable characteristics of the system like the number of losing
states. These experiments show no clear winner between $\Bw$ and
$\Fw$, which are rather complementary approaches. In terms of memory,
$\Bw$ outperforms $\Fw$ but $\Fw$ might run faster on some systems.

\section{Conclusion}
Inspired by Bonifaci and Marchetti-Spaccamela, we have considered a
game-based model for the online feasibility problem of sporadic tasks
(on multiprocessor platforms), and we have developed an efficient
algorithm to analyse these games, which takes into account the
specific structure of the problem.  Broadly speaking, we believe that
the application of well-chosen formal methods techniques (together
with specialised heuristics that exploit the structure of the
problems), is a research direction that has a lot of potential and
that should be further followed.

\input{main.bbl}

\end{document}

%% file: main.bbl
\newcommand{\etalchar}[1]{$^{#1}$}